\documentclass[12pt]{article}
\pdfoutput=1
\usepackage[utf8]{inputenc}
\usepackage[T1]{fontenc}
\usepackage[english]{babel}
\usepackage[margin=1.25in]{geometry}
\usepackage{amsmath}
\usepackage{amssymb, amsthm, mathrsfs}
\usepackage{mathtools, nccmath}
\usepackage{bbm}
\usepackage{setspace}
\usepackage{hyperref}
\usepackage{pdfpages}
\usepackage{graphicx}
\usepackage{subfig}
\usepackage{caption}
\usepackage{verbatim}
\usepackage{multirow}
\usepackage{rotating}
\usepackage{tabularx}
\usepackage{xcolor}
\usepackage{array}
\usepackage{makecell}
\usepackage{natbib}
\usepackage[capitalize]{cleveref}
\usepackage{enumerate}
\usepackage{thm-restate}
\usepackage{csquotes}
\usepackage[capposition=top]{floatrow}
\usepackage{soul}
\usepackage{tikz}
\usetikzlibrary{patterns}
\usepackage{subcaption}
\usetikzlibrary{calc}
\usetikzlibrary{decorations.pathreplacing}
\usetikzlibrary{arrows}

\newcommand{\given}{\,|\,}

\newcommand{\prob}[2][]{\text{\bf Pr}\ifthenelse{\not\equal{}{#1}}{_{#1}}{}\!\left[{\def\givenn{\middle|}#2}\right]}
\newcommand{\expect}[2][]{\text{\bf E}\ifthenelse{\not\equal{}{#1}}{_{#1}}{}\!\left[{\def\givenn{\middle|}#2}\right]}

\newcommand{\tparen}{\big}
\newcommand{\tprob}[2][]{\text{\bf Pr}\ifthenelse{\not\equal{}{#1}}{_{#1}}{}\tparen[{\def\given{\tparen|}#2}\tparen]}
\newcommand{\texpect}[2][]{\text{\bf E}\ifthenelse{\not\equal{}{#1}}{_{#1}}{}\tparen[{\def\given{\tparen|}#2}\tparen]}

\newcommand{\sprob}[2][]{\text{\bf Pr}\ifthenelse{\not\equal{}{#1}}{_{#1}}{}[#2]}
\newcommand{\sexpect}[2][]{\text{\bf E}\ifthenelse{\not\equal{}{#1}}{_{#1}}{}[#2]}

\newcommand{\lbr}[1]{\left\{#1\right\}}
\newcommand{\rbr}[1]{\left(#1\right)}

\newcommand{\dd}{{\,\mathrm d}}

\newcommand{\abs}[1]{\left|#1\right|}
\newcommand{\indicate}[1]{{\bf 1}\left[#1\right]}
\newcommand{\reals}{\mathbb{R}}
\newcommand{\posreals}{\reals_+}

\theoremstyle{definition}
\newtheorem{definition}{Definition}
\newtheorem{assumption}{Assumption}
\newtheorem{example}{Example}

\theoremstyle{plain}
\newtheorem{theorem}{Theorem}

\newtheorem{lemma}{Lemma}
\newtheorem{proposition}{Proposition}

\theoremstyle{remark}

\newtheorem{claim}{Claim}

\usepackage[nocomma]{optidef}

\graphicspath{ {./images/} }

\newcommand{\alloc}{Q}
\newcommand{\allocs}{\boldsymbol{Q}}
\newcommand{\signalalloc}{x}
\newcommand{\signalallocs}{\boldsymbol{x}}
\newcommand{\mechalloc}{y}
\newcommand{\mechallocs}{\boldsymbol{y}}
\newcommand{\contestalloc}{q}
\newcommand{\contestallocs}{\boldsymbol{q}}
\newcommand{\efficientalloc}{\alloc_{\rm E}}

\newcommand{\allocspace}{X}

\newcommand{\type}{\theta}
\newcommand{\types}{\boldsymbol{\theta}}
\newcommand{\typespace}{\Theta}
\newcommand{\Typespace}{\boldsymbol{\typespace}}
\newcommand{\reporttype}{\hat{\theta}}
\newcommand{\reporttypes}{\boldsymbol{\hat{\theta}}}
\newcommand{\signal}{s}
\newcommand{\signals}{\boldsymbol{\signal}}
\newcommand{\signalspace}{S}
\newcommand{\Signalspace}{\boldsymbol{\signalspace}}
\newcommand{\reportsignal}{\hat{s}}

\newcommand{\generalsignal}{\tilde{s}}

\newcommand{\util}{U}
\newcommand{\utils}{\boldsymbol{U}}
\newcommand{\expostutil}{u}

\newcommand{\dist}{F}
\newcommand{\dists}{\boldsymbol{F}}

\newcommand{\cost}{c}

\newcommand{\dista}{G}

\newcommand{\signalrecommend}{\tilde{\signal}}
\newcommand{\signalrecommends}{\tilde{\signals}}

\newcommand{\primed}{^\dagger}

\newcommand{\ranking}{r}
\newcommand{\tie}{z}
\newcommand{\ce}{\Sigma}

\providecommand{\keywords}[1]
{
  \small  
  \textbf{\textit{Keywords---}} #1
}
\providecommand{\JEL}[1]
{
  \small  
  \textbf{\textit{JEL---}} #1
}

\title{Mechanism Design under Costly Signaling:\\ the Value of Non-Coordination\thanks{We thank Asher Wolinsky, Bruno Strulovici, and Wojciech Olszewski for advice, helpful conversations, and comments. We also thank Ian Ball, Dirk Bergemann, Alex Bloedel, Eddie Dekel, Piotr Dworczak, Jeff Ely, Kira Goldner, Carl-Christian Groh, Yingni Guo, Marina Halac, Andrei Iakovlev, Annie Liang, Bart Lipman, Brendan Lucier, Deniz Kattwinkel, Joshua Mollner, Kyohei Okumura, Alessandro Pavan, Marcin Pęski,  Abhishek Sarkar, James Schummer,  Philipp Strack, Matthew Thomas, Yiyao Zhu, and participants at the UCLA Theory Seminar, 34th Stony Brook International Conference on Game Theory, Australasian Economic Theory Week, NUS Theory Lunch for helpful suggestions, comments, and discussions. Yingkai Li acknowledges financial support from NUS Start-up Grant. Part of the work was done while Yingkai Li was a Postdoc at Yale University under the support of the Sloan Research Fellowship, grant no.~FG-2019-12378.}}
\author{Yingkai Li\thanks{Department of Economics, National University of Singapore.
Email: \texttt{yk.li@nus.edu.sg}}
\and
Xiaoyun Qiu\thanks{Department of Economics, Dartmouth College.
Email:~\texttt{xiaoyun.qiu@dartmouth.edu}}}

\date{}

\begin{document}
\onehalfspacing

\maketitle

\begin{abstract}
We study allocation mechanisms that utilize costly signaling as a screening tool. A social planner aims to maximize social welfare, defined as the weighted sum of agents' utilities, while implementing a specific allocation rule. Within a broad class of agent preferences, we show that coordination mechanisms (where recommended signals depend on joint reports) are outperformed by non-coordination mechanisms (where signals depend solely on individual reports). When agents are also symmetric, such mechanisms are implementable through coarse-ranking contests. 
\end{abstract}


\JEL{D47, D61, D82}

\section{Introduction}
\label{sec:intro}
Economic resources are often allocated based on signals that are socially wasteful. We are interested in situations where the resources are owned by a central planner and use mechanism design to study welfare maximization.
A planner who can communicate with all agents can potentially improve welfare by coordinating their signaling, for example by advising likely losers to refrain from effort. However, such coordination requires credibly identifying the winner, which may necessitate a more stringent, costly signal from that agent. Coordination thus reduces waste on one margin but increases it on another. We study this tradeoff and ask whether coordination benefits welfare. The money-burning literature studies a similar mechanism design problem but typically assumes quasi-linear preferences \citep[e.g.,][]{hartline2008, condorelli2012money}. Under quasi-linearity, the payoff equivalence result of \citet{myerson1981} implies that it is without loss of optimality to use \emph{non-coordinated} (all-pay) mechanisms. However,  costly-signaling applications usually involve type-dependent costs that violate quasi-linearity.  We study environments in which payoff equivalence may fail and ask: when is optimal welfare achievable via non-coordinated mechanisms?

We consider a broad class of preferences that captures heterogeneous valuations and diverse signaling technologies,
including classic sorting conditions \citep{spence1973job}, designable costs \citep{chakravarty2013optimal}, and type-independent costs \citep{yang2024comparison}.
A prominent class of signaling technologies covered by our model is one where agents can reveal their true type for free but must incur private costs to inflate their signal \citep{perez2022test}.
These signaling technologies encompass real-world frictions, such as lying cost in public housing allocation, where applicants either truthfully present financial evidence or exert effort to underreport income to meet eligibility thresholds \citep{gao2005}; head starts in college admissions, where students with greater talent or access to tutoring resources exert less effort to achieve the same standardized test scores as their peers; and financing costs in performance-based procurement, where a firm must pay to issue debt above its current cash holdings to bolster its competitive position \citep{che1998standard}.

Our main result identifies a crucial condition for the optimality of non-coordinated mechanisms: Increasing Certainty Equivalence (ICE). ICE requires that   given any stochastic signal recommendation, its certainty equivalence, i.e., the deterministic signal that gives the agent the same expected utility, is weakly higher for higher types. Under ICE, we show that welfare-maximizing mechanisms are always non-coordinated (\cref{thm:monotone optimal general}).

Our result offers two insights.
 First, it serves as a robustness result: it generalizes the classic optimality of all-pay implementations from quasi-linear settings, where payoff equivalence holds, to a much broader class of preferences. This suggests that this optimality is a fundamental feature of preference structures rather than a byproduct of payoff equivalence.
Second, when a strict version of ICE holds, we show that coordination mechanisms are strictly sub-optimal. The cost of maintaining incentive compatibility in a coordination mechanism strictly exceeds the welfare saved by reducing waste. From this perspective, our results suggest that despite the potential for coordination to reduce wasted effort, letting agents compete publicly without designer intervention minimizes incentive costs and maximizes welfare. 
This provides a simplicity result for policymakers: in many realistic settings, there is no value in seeking elaborate direct mechanisms; optimal welfare is achieved through simple, non-coordinated competition.

The primary technical challenge in our analysis is comparing coordination and non-coordination mechanisms. Because payoff equivalence does not necessarily hold in our setting, coordination and non-coordination mechanisms implementing the same allocation rule can yield different utility levels for agents, even when the lowest type's utility is held constant. To analyze coordination mechanisms, we adopt a constructive approach. In a coordinating
mechanism, a recommendation may depend on all agents' reports and therefore be random
from an individual agent's interim perspective, even when it is deterministic conditional on
the full report profile. In a non-coordination mechanism, each recommendation depends deterministically only on that agent's own report.
We prove that for any stochastic mechanism, we can construct a deterministic counterpart that implements the same allocation while increasing every agent's utility under ICE. This construction demonstrates that the waste saved by coordination is always outweighed by the increased cost of ensuring truth-telling in a stochastic environment.

To illustrate how coordination impacts welfare and our core constructive idea, we compare a winner-pays mechanism (which coordinates signals based on joint reports) with an all-pay mechanism (which does not) in the following simplified setting (generalized in \cref{sec:model}). 

\begin{example}\label{ex:intro}
A principal allocates a single indivisible item to two agents based on public signals.
Each agent's private type $\type_i$ is drawn i.i.d. from a discrete uniform distribution on $\{L=0,H=\frac{1}{2}\}$. Agent~$i$ can generate a signal $\signal_i \in [0,\infty)$ at zero cost if $\signal_i \leq \type_i$; however, inflating the signal above their type ($\signal_i > \type_i$) incurs a quadratic cost $\cost_i = \frac{1}{2}(\signal_i - \type_i)^2$. With a common valuation of 1 for the item, an agent's utility is  $\expostutil_i = \signalalloc_i - \cost_i$, where $\signalalloc_i$ represents the allocation probability. The principal awards the item to the higher type, using a uniform tie-break. We compare two mechanisms that implement this allocation rule.

\textbf{Winner-pays (Coordinated) Mechanism:}
Each agent reports a type. Based on the reports, a potential winner is chosen as the one who reports a higher type under uniform tie-breaking. 
To maximize welfare,  type $L$ is recommended to play signal $0$.
Moreover, to deter type $L$ from mimicking type $H$, type $H$ is recommended to play signal $\frac{2}{\sqrt{3}}$ \emph{only if} they win.\footnote{Under this construction, type $H$ is strictly worse off by mimicking type $L$. We highlight that to maximize welfare, the upward-deviation constraint binds, whereas in the classical revenue-maximization auction design problems, the downward-deviation constraint binds. 
We will discuss this contrast in greater detail in \cref{sec:main theorem}.}
The expected utility is $\frac{1}{4}\cdot(1 -0) = \frac14$ for  type $L$ and is  $\frac{3}{4}\cdot(1 - \frac{1}{2}(\frac{2}{\sqrt{3}}-\frac{1}{2})^2) \approx 0.589$ for type  $H$.
Intuitively, the winning agent must burn utility to credibly signal superiority.

\textbf{All-pay (Uncoordinated) Mechanism:} Type $L$ is recommended to play a deterministic signal $0$, and type $H$ is recommended to play a deterministic signal $1$, regardless of whether they win the item or not.\footnote{In this mechanism, type $H$ is strictly worse off by mimicking type $L$, while the upward-deviation constraint is binding: the payoff of type $L$ from mimicking type $H$ is $\frac{3}{4} - \frac{1}{2}(1-0)^2 = \frac{1}{4}$, conditional on the other agent playing the equilibrium strategy.}
The item is allocated to the type with a higher signal, with ties broken uniformly randomly. 
While the expected utility for type $L$ remains unchanged, the expected utility for type $H$ is 
$\frac{3}{4} - \frac{1}{2}(1-\frac{1}{2})^2 =0.625 > 0.589$, strictly higher than that in the winner-pays mechanism.

The key distinction between these two mechanisms is that recommendations are deterministic functions of each agent's own type in the all-pay mechanism, while in the winner-pays mechanism, they are stochastic from each agent's interim perspective, where the randomization arises from the randomization of others' reports. The latter coordinates signals so that only the winner expends effort, saving the loser's cost. However, this coordination raises incentive costs: the winner's signal must be sufficiently high to deter imitation.

To be more concrete, type $H$ receives a stochastic signal recommendation in the winner-pays mechanism: with probability $\frac34$, they are the potential winner with a signal requirement $\frac{2}{\sqrt{3}}$; with probability $\frac14$, they are picked as the loser with no signal requirement. 
Let $\bar{s}$ and $\underline{s}$ be the deterministic signals that yield the same utility as this stochastic lottery for types $H$ and $L$, respectively.
Under ICE, which is satisfied by the quadratic cost, $\bar{s}>\underline{s}$.\footnote{The high and low types' certainty-equivalent signals $\bar{s}$ and $\underline{s}$ satisfy $\frac12(\bar{s}-\frac12)^2 = \frac34\cdot\frac12 (\frac{2}{\sqrt{3}}-\frac12)^2$ and $\frac12(\underline{s}-0)^2 = \frac34\cdot\frac12 (\frac{2}{\sqrt{3}}-0)^2$, respectively. We therefore have $\bar{s}=\frac{3}{2}-\frac{\sqrt{3}}{4}\approx 1.07$ and $\underline{s}=1$.}

If we replace the winner-pays lottery with a deterministic signal $\bar{s}$ for type $H$, type $H$'s utility remains unchanged. However, because $\bar{s}>\underline{s}$, this deterministic signal is more costly for a low type to mimic than the original lottery was.
This slackness in incentive compatibility allows the principal to lower the signal requirement for type $H$ while still deterring type $L$. Consequently, the all-pay mechanism reduces the total waste required to maintain truth-telling, leading to higher social welfare.

\end{example}

The observation that removing coordination improves welfare is not contingent on the quadratic cost structure or a specific allocation rule. Rather, the fundamental requirement is the ICE condition. This condition guarantees that a non-coordination mechanism can mitigate (upward) deviation incentives under monotonic allocation rules, which lead to welfare improvement.
Our main theorem (\cref{thm:monotone optimal general}) formalizes this intuition by demonstrating that for any \emph{monotone} interim allocation, coordinated signal recommendations are dominated by non-coordination mechanisms. While non-monotone allocation rules are often infeasible in practice due to fairness concerns, we establish in \cref{sec:discussions} that even when such allocations are implementable, they are generally dominated by monotone allocations under mild conditions.

Our central contribution is the observation that non-coordination can be optimal in environments where payoff equivalence fails. Although \citet{perez2024score} study a quadratic cost structure similar to that in \cref{ex:intro}, their analysis is restricted to a single-agent model, thus bypassing the coordination problem entirely.
Our findings also depart from the money-burning literature, which typically relies on quasi-linear preferences to ensure payoff equivalence. By relaxing these assumptions, we identify a novel economic tradeoff inherent in coordination. This framework allows us to address a new design question: whether complex coordination mechanisms are necessary for welfare maximization. Finally, while \citet{zhang2023optimal} demonstrates the optimality of all-pay mechanisms for effort maximization under separable convex costs, our model focuses on welfare and accommodates a broader class of type-dependent and non-convex cost functions.

The remainder of the paper is organized as follows. In the next subsection, we relate our work to the existing literature. In \cref{sec:model}, we introduce the general model. In \cref{sec:non_coordination_is_optimal}, we present our main result and the intuition behind it. In \cref{sub:implementation}, we illustrate the connections between non-coordination (all-pay)  mechanisms and traditional ranking-based contests by proposing coarse ranking contests as practical implementations of non-coordination mechanisms in symmetric environments. In \cref{sec:discussions}, we present a partial converse result on when coordination can improve welfare, and generalize the optimality of non-coordination mechanisms to non-monotone allocations. 
In \cref{sec:conclusion}, we conclude with directions for future research.

\subsection{Related Work}\label{sec:review}

Our paper is closely related to the literature on money burning and costly ordeals. 
An early precursor to this literature is the work of \citet{mcafee1992bidding}, who study bidding rings without side payments and show that in optimal collusion, bidders with values above the reserve price optimally submit identical bids at that price. This and subsequent work typically assumes that agents' preferences are quasi-linear in transfers \citep[][]{hartline2008,condorelli2012money,finkelstein2019take} or in effort-based costs independent of private information \citep[][]{chakravarty2013optimal, yang2024comparison}. 
We depart from this literature by relaxing the quasi-linearity assumption. Under quasi-linearity, the mechanism design space simplifies: for any mechanism that coordinates agents' signaling choices, there exists a non-coordinating counterpart that yields identical utility (welfare) for each agent. In our more general setting, mechanisms implementing the same allocation rule can result in disparate welfare outcomes. While this expansion of the feasible mechanism space introduces new technical challenges, it allows us to uncover a novel economic tradeoff regarding the optimality of coordination. 
We introduce the Increasing Certainty Equivalence (ICE) assumption (\cref{asp:monotone_certainty_equivalence}), under which we prove that welfare-maximizing mechanisms feature no coordination.
We thus provide a robust generalization of the optimality of non-coordination mechanisms to a broader class of preferences. 

Our paper is also related to mechanism design beyond quasi-linearity, where transfers generate wealth effects through risk aversion or financial constraints. In such environments, payoff-equivalence arguments can fail and the details of implementation can matter for welfare and incentives \citep[e.g.,][]{maskin1984optimal,matthews1987comparing,che1998standard}. More broadly, recent work studies incentive compatibility and mechanism design without quasi-linearity under general nonlinear preferences \citep[e.g.,][]{kosMessner2013nonquasi,kazumuraMishraSerizawa2020nonquasi}. Our setting differs in that monetary transfers are replaced by costly (possibly type-dependent) signals, and the principal seeks to maximize agents' utilities. Nonetheless, the shared theme is that nonlinearity enlarges the relevant design space; our ICE condition restores tractability by implying that non-coordination mechanisms weakly dominate coordinating ones for implementing any monotone allocation.

Our paper is also conceptually related to the literature on falsification and costly lying. \citet{green1986partially} study auction settings where lying costs are either zero or infinite, while later papers consider richer cost structures but focus on different objectives, such as classification \citep{hardt2016strategic} or test design \citep{perez2022test}. 
The most closely related work is \citet{perez2024score}. They show that under linear or quadratic falsification costs, the optimal mechanism is score-based; that is, it recommends a deterministic action based on the report. The main difference is that they study a single-agent model, ruling out the possibility of coordination. In contrast, we study a multi-agent environment. This allows us to study the welfare consequences of coordinating agents' signaling choices. 

The costly signaling aspect of our model is rooted in the classic signaling literature \citep[e.g.,][]{spence1973job} and the more recent work on strategic gaming of signals \citep[e.g.,][]{FK2019muddled,Ball2019}. The main difference is that, unlike standard signaling games where markets reward agents based on beliefs about their type, we study institutional settings where a principal directly controls allocation and can design mechanisms under resource constraints. This design perspective enables us to analyze how different implementations of the same allocation rule lead to different agent behaviors and social costs.

Our non-coordination mechanisms can also be implemented as coarse ranking contests, which generalize traditional contest formats such as all-pay auctions \citep{baye1996all}, tournaments \citep{lazearRosen1981rank}, and Tullock contests \citep{skaperdas1996contest}. Coarse rankings allow the mechanism to map multiple performances into the same tier, enabling robustness and reducing unnecessary effort. \citet{zhang2023optimal} also shows that contests can be optimal for inducing effort, under special signaling costs. Our key departure from this literature is that we focus on minimizing socially wasteful effort. 

\section{Model}
\label{sec:model}
The principal (she) wishes to allocate $k$ identical items to~$n> k$ heterogeneous agents (he).
An allocation $\signalallocs=(\signalalloc_i)_{i=1}^n$ is a vector of probabilities such that $0\leq\signalalloc_i\leq 1$ for each $i$, and $\sum_{i=1}^n\signalalloc_i\leq k$.\footnote{Our paper focuses on the interpretation where items are indivisible and $x_i$ represents allocation probabilities. All the results extend naturally by considering divisible items and interpreting $x_i$ as fractional allocations.}
Let $\allocspace \subseteq [0,1]^n$ be the space of feasible allocations. 

Each agent $i$ has a private type~$\type_i$ drawn independently from a publicly known distribution~$\dist_i$ supported on $\Theta_i=[\underline{\type}_i,\bar{\type}_i]\subseteq \reals_+$.
The distribution of the type profile $\types$ is~$\dists$, whose support is $\Typespace = \prod_{i=1}^n\typespace_i$.
Each agent $i$ can generate a costly public signal $\signal_i\in \signalspace=\reals_+$ 
and the utility of the agent is $\expostutil_i(\signalalloc_i,\signal_i,\type_i)$ when receiving an item 
with probability $\signalalloc_i\in[0,1]$ and generating a signal $\signal_i$.
We assume that $\expostutil_i(\signalalloc_i,0,\type_i) \in [0,1]$ for any $\signalalloc_i,\type_i$. 
That is, all agents have a weakly positive value for any allocation given their types when not generating any costly signal, and this value is bounded.  
Denote the Cartesian product of agents' signal space by $\Signalspace = \signalspace^n$.
The principal cannot observe agents' private types, but she can observe their signals. 

\begin{assumption}[regularity]\label{asp:monotonicity}
For any agent $i$, his utility function $\expostutil_i$ is continuous in all of its coordinates, bounded in value, strictly increasing in $x_i$, weakly decreasing in $s_i$, and weakly increasing in $\type_i$. Moreover, it satisfies the von Neumann-Morgenstern expected utility representation.
\end{assumption}

Under \cref{asp:monotonicity}, each agent always strictly prefers a higher probability of allocation. Moreover, since it is costly to generate any signals, the cost of signaling is a  weakly increasing function of signal production. 
There are two potential reasons why the utility is an increasing function of type. 
First, a higher type may value the item more. 
Second, a higher type may have a weakly lower cost of generating the same signal. 
Our assumption is general enough to cover both scenarios. 

\paragraph{Mechanisms.}
We want to define a class of mechanisms as general as possible.
For example, we want to incorporate the possibility of communication between the principal and agents, and the choice of the mechanisms can be contingent on such communication.
Following \citet{myerson1982optimal}, we know that it is both convenient and without loss of generality to study direct mechanisms.\footnote{In \citet{myerson1982optimal}, such mechanisms are called direct coordination mechanisms. We call them direct mechanisms or general mechanisms, because some of them exhibit zero coordination.} 
In these direct mechanisms, we allow the principal to receive type reports from all agents, based on which she makes signal recommendations. Moreover, the allocation of items can be contingent on both the reports and the observed signals.

Formally, the principal first commits to a signal recommendation policy $\signalrecommends:\Typespace\to \Delta(\Signalspace)$ and an allocation rule $\mechallocs:\Typespace\times\Signalspace\to \allocspace$.
After each agent $i$ reports their type $\reporttype_i$, the principal uses  the reported type profile $\reporttypes$ to draw a signal recommendation profile $\tilde{\signals}\sim\signalrecommends(\reporttypes)$ and sends a recommendation $\tilde{\signal}_i$ to each agent $i$. 
Upon receiving this recommendation, each agent produces signal $\signal'_i\in \signalspace$.
The principal then observes the realized signal profile $\signals'$ and awards an item to each agent $i$ with probability 
\[
\indicate{\signal'_i=\tilde{\signal}_i}\,\mechalloc_i(\reporttypes,\signals').
\]
This formulation explicitly imposes a compliance check: an agent whose signal deviates from their recommendation receives no item. This enforcement convention preserves the intended allocations when agents comply and can only penalize a unilaterally noncomplying agent. Furthermore, because signaling is costly, an agent who disobeys their recommendation can do no better than choosing signal zero. We may therefore restrict each agent's choice to $\signal'_i\in\{\tilde{\signal}_i,0\}$.\footnote{We impose interim participation constraints for agents, i.e., agents can choose to walk away after seeing the signal recommendation. An alternative formulation is to only consider the ex ante participation constraint where each agent only has the option to opt out at the beginning of the mechanism, and they are forced to follow the principal's recommendation after participation.  This distinction is not crucial for our analysis, and the main results of the paper hold for both cases.}
We denote this direct mechanism by $(\signalrecommends,\mechallocs)$.

We highlight a special class of direct mechanisms that will be of importance in our analysis.
These are direct \emph{non-coordination} mechanisms. 
They impose the following additional restrictions on direct mechanisms: (1) the signal recommendation policy for each agent $i$ is a function of his type, that is,  $\signalrecommend_i:\Theta_i\to \signalspace_i$, and (2) the allocation rule can only depend on public signals, that is, $\mechallocs:\Signalspace\to \allocspace$.
These restrictions rule out the possibility of using reports from all agents to coordinate agents' choices of signals. 

It is worth noting that direct non-coordination mechanisms can be implemented as \emph{signal-based} mechanisms. In a signal-based mechanism, the principal commits to a signal-based allocation rule $\signalallocs: \Signalspace \rightarrow \allocspace$, which maps each realized signal profile to a (randomized) allocation. 
After that, each agent~$i$ chooses a signal strategy $\signal_i:\Theta_i\to\Delta(\signalspace)$ to maximize his expected utility. 
Given the strategy profile $\signals=(\signal_i)_{i=1}^n$ and the realized signal profile $\signals(\types)=(\signal_i(\type_i))_{i=1}^n$, an item is allocated to agent $i$ with probability~$\signalalloc_i(\signals(\types))$.
For any direct non-coordination mechanism, it is straightforward to show that there exists a signal-based mechanism that implements the same outcome under a pure Nash equilibrium.\footnote{A similar observation on this indirect implementation is also made in the contemporary work by \citet{perez2024score} in a single-agent environment. } 

\paragraph{Interim approach.}
Given any direct mechanism $(\signalrecommends,\mechallocs)$, the \emph{interim allocation} and the \emph{interim utility} of agent $i$ with private type $\type_i$ are 
\begin{align}\label{interim alloc}
  \alloc_i(\type_i)=\expect[\types_{-i}]{\expect[\signals\sim\signalrecommends(\type_i,\types_{-i})]{\mechalloc_i(\type_i,\types_{-i},\signals)}}  
\end{align}
and   
\begin{align}\label{interim utility}
    \util_i(\type_i) = \expect[\types_{-i}]{\expect[\signals\sim\signalrecommends(\type_i,\types_{-i})]{\expostutil_i(\mechalloc_i(\type_i,\types_{-i},\signals),\signal_i,\type_i)}}
\end{align}
respectively.
 
A direct mechanism $(\signalrecommends,\mechallocs)$ is incentive compatible (IC) if, for any agent $i$ and any types $\type_i,\type'_i$, 
his expected utility from truthfully reporting his type and following the signal recommendation is not lower than that from misreporting as $\type'_i$ and deviating from the realized signal recommendation $\signal_i\sim\signalrecommend_i(\type'_i, \types_{-i})$ according to any function $s'_i(\signal_i) \in \{0,\signal_i\}$. That is,
\begin{align*}
&\util_i(\type_i) \geq \expect[\types_{-i}]{\expect[\signals\sim \signalrecommends(\type'_i, \types_{-i})]{
\expostutil_i(\indicate{\signal'_i(\signal_i)=\signal_i}\,\mechalloc_i(\type'_i,\types_{-i},\signal'_i(\signal_i),\signals_{-i})
,\signal'_i(\signal_i),\type_i)}}.\tag{IC}
\end{align*}
We denote the interim allocation profile and interim utility profile as $\allocs=(\alloc_i)_{i=1}^n$ and $\utils=(\util_i)_{i=1}^n$, respectively.

\begin{definition}[implementability]\label{def:implementability}
An interim allocation--utility profile $(\allocs,\utils)$ can be \emph{implemented by a direct mechanism}
if there exists a signal recommendation policy~$\signalrecommends:\Typespace\to \Delta(\Signalspace)$ and an
allocation rule~$\mechallocs:\Typespace\times\Signalspace\to \allocspace$ such that (i) the direct mechanism $(\signalrecommends,\mechallocs)$ is incentive compatible, and (ii) the consistency conditions \eqref{interim alloc} and \eqref{interim utility} hold for any agent $i$ and type $\type_i$. 

An interim allocation $\allocs$ can be \emph{implemented} (by a direct mechanism) if there exists an interim utility profile $\utils$ such that $(\allocs,\utils)$ can be implemented by a direct mechanism.
\end{definition}

Analogously, an interim allocation--utility profile $(\allocs,\utils)$ can be \emph{implemented by a direct non-coordination mechanism} if the direct mechanism $(\signalrecommends,\mechallocs)$ in \cref{def:implementability} is a direct non-coordination mechanism. That is, the signal recommendation policy for each agent $i$ only depends on his own type: $\generalsignal_i:\Theta_i \rightarrow \signalspace$, and the allocation rule only depends on the public signals: $\mechallocs: \Signalspace \rightarrow \allocspace$.   The consistency condition now becomes
\begin{align*}
    \tag{C'}
    \alloc_i(\type_i)&=\expect[\types_{-i}]{\mechalloc_i(\signalrecommend_i(\type_i),\signalrecommends_{-i}(\types_{-i}))},\\
    \util_i(\type_i) &=
\expect[\types_{-i}]{\expostutil_i(\mechalloc_i(\signalrecommend_i(\type_i),\signalrecommends_{-i}(\types_{-i})),\signalrecommend_i(\type_i),\type_i)}.
\end{align*}

Notice that in a direct mechanism, both the allocation and the recommended signal could be stochastic from each agent's perspective, whereas in a direct non-coordination mechanism, only the allocation could be stochastic from each agent's perspective. We introduce two more simplifying notations. For any agent $i$,  denote the expected utility of type $\type$ by $\util_i(\type;G):= \expect[(x_i,s_i)\sim G]{\expostutil_i(x_i,\signal_i,\type_i)}$ when the allocations and recommended signals follow the distribution $G$. Denote the expected utility of type $\type$ by $\util_i(\type;G_{\allocspace},\signal)$ when the allocations follow the distribution $G_{\allocspace}$ and the signal is $\signal\in \posreals$. 

\paragraph{Principal's payoff.} 
The principal is a social planner who cares about maximizing the utilities of agents when implementing any targeted allocation rule. The principal relies on public signals to screen agents, but is aware that excessive signaling is socially wasteful.
 Formally, given any allocation rule $\allocs$ that can be implemented by a direct mechanism, the main objective is to design a mechanism that implements $\allocs$ and maximizes the weighted utilities of the agents, that is,  
\begin{align*}
\sum_{i=1}^n \expect[\type_i\sim\dist_i]{w_i(\type_i)\cdot \util_i(\type_i)},
\end{align*}
where $w_i(\type_i)\in[0,1]$ is a weighting function that can depend on each agent's private type.\footnote{We formulate this objective broadly in terms of utility because our general model does not assume additive separability. However, in the special case where utility is additively separable, $\expostutil_i(x_i,s_i,\type_i)=v_i(x_i,\type_i)-c_i(s_i,\type_i)$, fixing the interim allocation rule inherently fixes the expected value derived from the allocation. In that specific context, maximizing the displayed objective becomes mathematically equivalent to minimizing the weighted expected signaling costs.}

\section{Optimality of Non-coordination Mechanisms}
\label{sec:non_coordination_is_optimal}
In this section, we show that under some conditions, given any monotone allocation rule, the optimal mechanism is a non-coordination mechanism.
\subsection{Assumptions}\label{sec:assumptions}
We first introduce the standard weak single-crossing property as in \citet{milgrom1994monotone}.
\begin{assumption}[weak single-crossing property (wSCP)]\label{asp:weak_single_crossing}
For any agent $i$, his utility function $\expostutil_i$ satisfies the \emph{weak single-crossing property (wSCP)}: for any types $\type_i'>\type_i$, any signals $\signal_i'\geq\signal_i$ and any allocation probabilities $\signalalloc_i'\geq\signalalloc_i$, 
\begin{enumerate}
    \item $\expostutil_i(\signalalloc_i',\signal_i',\type_i)-\expostutil_i(\signalalloc_i,\signal_i,\type_i) > 0$ 
implies that $\expostutil_i(\signalalloc_i',\signal_i',\type_i')-\expostutil_i(\signalalloc_i,\signal_i,\type_i') > 0$;
\item $\expostutil_i(\signalalloc_i',\signal_i',\type_i)-\expostutil_i(\signalalloc_i,\signal_i,\type_i) \geq 0$ 
implies that $\expostutil_i(\signalalloc_i',\signal_i',\type_i')-\expostutil_i(\signalalloc_i,\signal_i,\type_i') \geq 0$.
\end{enumerate}
\end{assumption}

Note that \cref{asp:weak_single_crossing} only requires comparisons across deterministic outcomes. This restriction is important in our setting because, in a general direct mechanism, an agent may face stochastic signal recommendations, especially when the principal coordinates agents' signaling choices using cross-agent information. 
As a result, \cref{asp:weak_single_crossing} alone is typically not sufficient to ensure that local incentive constraints characterize global incentive compatibility in arbitrary direct mechanisms.\footnote{\citet{kartik2024single} study single-crossing differences (SCD), which impose single-crossing comparisons for arbitrary lotteries. This requirement is stronger than \cref{asp:weak_single_crossing} and therefore applies to a narrower class of utility functions. For example, $u(x,\signal,\type)=x-(\signal-\type)^+$ satisfies \cref{asp:weak_single_crossing} but not SCD.}
To circumvent this technical challenge, we adopt a constructive approach: starting from any direct mechanism that implements a monotone allocation rule, we explicitly construct a non-coordination mechanism that implements the same allocation rule and weakly improves agents' utilities. This construction relies crucially on the following assumption.

\begin{definition}\label{def:certainty_equivalence}
Suppose allocations and signals follow the distribution $G$. Denote the marginal distribution over allocations by $G_{\allocspace}$. The certainty equivalent signal $\ce_i(G,\type)$ of the distribution $G$ for type $\type$ satisfies the following condition: $\util_i(\type;G_{\allocspace},\ce_i(G,\type))=\util_i(\type;G)$.
\end{definition}

\begin{assumption}[Increasing Certainty Equivalence]
\label{asp:monotone_certainty_equivalence}
A utility function $\expostutil_i$ is said to have increasing certainty equivalence if for any distribution $G$ over allocations and signals, his certainty equivalent signal $\ce_i(G,\type)$ is increasing in $\type$. 
\end{assumption}

\Cref{asp:monotone_certainty_equivalence} requires that higher-type agents are able to produce higher certainty-equivalent signals for any given lottery over signal recommendations.
To illustrate this condition more concretely, consider a class of utility functions that are additively separable in valuation for the allocation and signaling cost
$$
\expostutil_i(\signalalloc_i,\signal_i,\type_i) = v_i(\signalalloc_i,\type_i)- c_i(\signal_i,\type_i).$$
Here, $v_i(\cdot,\cdot)$ represents the agent's valuation  and $c_i(\cdot,\cdot)$ denotes the cost of signaling.\footnote{In settings without additive separability, we do not have simple sufficient conditions for \cref{asp:monotone_certainty_equivalence}. While we focus on additive separability for interpretability, our results apply whenever \cref{asp:monotone_certainty_equivalence} holds. In \cref{apx:non_separable}, we extend the core logic of \cref{thm:monotone optimal general} to non-separable utilities by proposing a weaker version of \cref{asp:monotone_certainty_equivalence} and providing sufficient conditions for its validity.
} 

\paragraph{Type-independent and Multiplicatively Separable Costs} In this case, the agents' utility functions take the form $\expostutil_i(\signalalloc_i,\signal_i,\type_i) = v_i(\signalalloc_i,\type_i)- c_i(\signal_i)$. 
Provided that $c_i$ is strictly increasing, we can define $t_i=c_i(\signal_i)$ and treat the cost itself as the public signal. Relabeling the recommendations and allocation rules in terms of $t_i$ leaves all incentives unchanged and simplifies the utility to $v_i(\signalalloc_i,\type_i)-t_i$. Consequently, any strictly increasing, type-independent cost can be normalized to a linear disutility without loss of generality. This framework naturally accommodates models of mechanism design with money burning or costly ordeals, where an agent's utility typically takes the linear form $\expostutil_i(\signalalloc_i,\signal_i,\type_i) = \type_i \cdot \signalalloc_i - \signal_i$ \citep[e.g.,][]{hartline2008}. Furthermore, when viewed in the original signal units, this preference class encompasses a wide variety of cost structures studied in the contest literature, including cases where $c_i(\cdot)$ is convex \citep[e.g.,][]{fang2020turning,zhang2023optimal}, concave \citep[e.g.,][]{moldovanu2001optimal}, convex-then-concave or concave-then-convex \citep[e.g.,][]{kim2024choosing}, or even more complex forms.

More generally, our framework accommodates multiplicatively separable costs of the form: 
$c_i(\signal_i,\type_i) = h_i(\type_i)\cdot g_i(s_i)$.
An agent with these preferences behaves identically to one with type-independent costs under the normalized utility 
$\frac{v_i(\signalalloc_i,\type_i)}{h_i(\type_i)} - g_i(s_i)$. 
This allows us to capture settings in the redistribution literature
where the agent's private type encodes two-dimensional information (valuation $v_i$ and marginal cost $c_i$)  
 and utility is
$
\expostutil_i = v_i\cdot \signalalloc_i - c_i \cdot \signal_i$ \citep{akbarpour2024redistributive,yang2024comparison}.

In these environments, \cref{asp:monotone_certainty_equivalence} holds as a boundary case: the certainty-equivalent signal is constant across types because the conversion of a signal distribution into its deterministic equivalent depends solely on $c_i(\signal_i)$.
In these cases, payoff equivalence holds \citep{myerson1981, zhang2023optimal}: Any two mechanisms implementing the same interim allocation induce the same utilities for the agents.
Thus, the optimality of non-coordination follows trivially.

\paragraph{Type-dependent Costs}
Our framework also accommodates genuinely type-dependent cost functions where payoff equivalence fails.
While the classic signaling literature typically relies on the single-crossing property \citep{spence1973job}, \cref{asp:monotone_certainty_equivalence} imposes additional structure on the curvature of the cost function.

Technically, \cref{asp:monotone_certainty_equivalence} is equivalent to assuming that the disutility of signaling, $-c_i(\signal_i,\type_i)$, exhibits weakly decreasing absolute risk aversion (DARA) in type $\type_i$ for the cost of signaling \citep{pratt1976risk, maskin1984optimal}. If $\cost_i(\cdot,\cdot)$ is third-differentiable, \cref{asp:monotone_certainty_equivalence} is satisfied if and only if \footnote{In \citet{pratt1976risk,maskin1984optimal}, one often additionally assumes risk aversion, i.e., $-c_{ss}\leq 0$. We do not require this restriction for \cref{asp:monotone_certainty_equivalence}.}
\begin{align}\label{eq:risk_constant}
\frac{\partial}{\partial \type_i} \rbr{-\frac{c_{ss}(\signal_i,\type_i)}{c_s(\signal_i,\type_i)}} \leq 0, \,\,\forall \signal_i,\type_i,
\,\, \text{where } c_{ss}(\signal_i,\type_i) = \frac{\partial^2 c_i(\signal_i,\type_i)}{\partial \signal_i^2},
c_{s}(\signal_i,\type_i) = \frac{\partial c_i(\signal_i,\type_i)}{\partial \signal_i}.
\end{align}
\citet{maskin1984optimal} provide further simplifying assumptions on primitives (in their assumptions A and B) such that the above condition holds.

While this connection is technically interesting, we highlight that  ICE should be interpreted as a curvature restriction on signaling costs rather than a restriction on risk preferences over outcomes. Notice that we usually define ``risk'' over outcomes: $v_i-c_i$ if agent $i$ gets an item and $-c_i$ if he does not get an item. 
Since preferences are linear in the allocation, agents remain risk-neutral in the traditional sense. ICE simply implies that higher types are more capable of generating higher deterministic equivalent signals to resolve the randomness in recommendations.

A commonly used reduced-form specification assumes that only upward falsification is costly, taking the form
\[
c_i(\signal_i,\type_i)=C\!\big((\signal_i-\type_i)^+\big),
\]
for some increasing function $C(\cdot)$ with $C(0)=0$ \citep{perez2023fraud}.
This formulation captures the intuition that the agent pays only for the incremental ``gap'' between the submitted signal and the true type. On the positive-cost region $z=\signal_i-\type_i>0$, the Pratt condition in \eqref{eq:risk_constant} reduces to
\[
\frac{\mathrm d^2}{\mathrm d z^2}\log C'(z)\leq 0.
\]
Thus, under the usual differentiability conditions and $C'(z)>0$, concavity of $\log C'(z)$ is a sufficient condition for ICE; the kink at $z=0$ can be handled directly. The condition is satisfied, for example, by linear and quadratic $C$. The resulting linear upward-falsification cost $c_i(\signal_i,\type_i)=(\signal_i-\type_i)^+$ and the related unrestricted quadratic cost $c_i(\signal_i,\type_i)=(\signal_i-\type_i)^2$ are studied in \citet{perez2024score}.\footnote{While the unrestricted quadratic cost is not globally monotone in $\signal_i$, \citet{perez2024score} show that this is inconsequential for implementing monotone allocations, as equilibrium signal recommendations necessarily exceed the true type.}

Similar distance-based costs appear in contest models with ``head starts'',  where a baseline advantage shifts the feasible frontier of attainable scores. These head starts may reflect innate talent, signaling greater merit, or socioeconomic advantages, such as access to elite tutoring.  
This structure is explored in \citet{siegel2009all}, where each player $i$ chooses a score $s_i \in [a_i,\infty)$, and $a_i \ge 0$ is an initial score capturing the head start. By interpreting $e_i:=s_i-a_i\ge 0$ as incremental effort, we obtain the cost representation $c_i(s_i,a_i)=C((s_i-a_i)^+)$, which aligns directly with our framework.

Finally, \cref{asp:monotone_certainty_equivalence} can be interpreted as the cost of payments under financial constraints. In this context, it is useful to distinguish the baseline payment from the premium incurred when expenditures exceed available liquid funds.
\citet{che1998standard} model the cost of paying $x$ given budget $w$ via a cost-of-spending function $C(x,w)$, and study the cost function $C(x,w)=x+R((x-w)^+)$
where $R(\cdot)$ represents an increasing, convex borrowing or financing premium satisfying $R(0)=0$.
Mapping the payment $x$ to our signal $\signal_i$ and the budget $w$ to our type $\type_i$, we get a cost function
\[
c_i(\signal_i,\type_i)=\signal_i + R\big((\signal_i-\type_i)^+\big).
\]
Applying the criterion in \citet{pratt1976risk}, we find that this specification satisfies \cref{asp:monotone_certainty_equivalence} if $\log(1+R'(z))$ is concave, a condition met by both linear $R(z)=r z$ and the quadratic financial premia $R(z)=\kappa z^2$.

\subsection{Main Theorem}\label{sec:main theorem}

Now we are ready to state our main theorem.
\begin{theorem}[Optimality of Non-coordination]
\label{thm:monotone optimal general}
Suppose Assumptions~\ref{asp:monotonicity}, \ref{asp:weak_single_crossing}, and \ref{asp:monotone_certainty_equivalence} hold.
Fix any interim allocation--utility pair $(\allocs,\utils)$ such that $(\allocs,\utils)$ can be implemented by a direct mechanism and $\alloc_i$ is weakly increasing for all $i$. 
Then there exists an interim allocation--utility pair $(\allocs\primed,\utils\primed)$ that can be implemented by a direct non-coordination mechanism
such that $\allocs\primed=\allocs$, and 
$
\util\primed_i(\type_i) \geq \util_i(\type_i)$ for all $ i, \type_i$. 
\end{theorem}

We highlight that \cref{thm:monotone optimal general} can be strengthened to a strict version.
Consider the following \emph{strict} version of ICE: for any \emph{non-degenerate} distribution $G$ over allocations and signals, the certainty equivalent signal $\ce_i(G,\type)$ is strictly increasing in $\type$.\footnote{The strict ICE holds, for example, when the condition in \cref{eq:risk_constant} holds with strict inequalities. } 
Under this strict ICE condition, any coordinating mechanism that involves stochastic signal recommendations for some types is \emph{strictly} worse than the optimal non-coordination mechanism.

\cref{thm:monotone optimal general} delivers a sharp simplification: under ICE, for any implementable \emph{monotone} interim allocation, coordination in signal recommendations is never welfare-improving. 
Under the strict ICE condition, coordination is strictly sub-optimal.
This suggests that to maximize agents' utilities, there is no need to utilize cross-agent information to coordinate signaling choices (or ``leak'' others' information to each agent).
Coordination mechanisms typically involve stochastic signal recommendations, whose realizations depend on all agents' private information.
Under the ICE condition, such stochastic recommendations can be improved upon using deterministic signals, where each agent's signaling choice does not depend on others' private information.
This result suggests that the designer can safely use simple mechanisms like the non-coordination ones without sacrificing welfare.

A non-coordination mechanism admits an \textbf{all-pay} format. From this perspective, \cref{thm:monotone optimal general} suggests that the welfare-maximizing implementation admits an all-pay format when using costly signals as screening devices in allocation problems. 
At first glance, this may seem counterintuitive since all-pay formats require both winners and losers to incur signaling costs. 
However, we will show that all-pay formats can deter lower types from mimicking higher types more effectively. On average, each agent is required to produce a ``lower signal'' while maintaining incentive compatibility.
This allows every type to retain more information rent, improving utilities across the board.
Overall, all-pay formats have a lower ``incentive cost''.

\paragraph{Proof Sketch.} The proof is constructive and is provided in \cref{apx:main thm}.\footnote{This proof assumes a discrete type space. The extension to a continuous type space is provided in \cref{subapx:continuous}.} Here, we only outline the main steps.
Consider an arbitrary agent $i$ and an arbitrary type~$\hat{\type}_i$.
Fix an arbitrary direct mechanism $\mathcal{M}$ that induces the interim allocation and interim utilities $(\allocs,\utils)$, where the allocation and recommended signal for type~$\hat{\type}_i$ follow a distribution $G$.
We focus on the case where the distribution $G$ is non-degenerate so that $\mathcal{M}$ coordinates agents' signals. To simplify the exposition, consider a finite type space and let $\tilde{\type}_i$ be the largest type below $\hat{\type}_i$. 

Consider a candidate mechanism $\mathcal{M}\primed$ in which type $\hat{\type}_i$ is recommended to produce a deterministic signal $\hat{s}_i$ and receives the item following the marginal distribution $G_{\allocspace}$ induced by $G$, leaving the rest of the mechanism unchanged from the original mechanism. 
Let $\utils\primed$ be the interim utilities of the agents in mechanism $\mathcal{M}\primed$.
We choose $\hat{s}_i$ so that the upward-deviation constraint of the lower type $\tilde{\type}_i$ binds under the candidate mechanism, i.e., $\util\primed_i(\tilde{\type}_i)=\util\primed_i(\tilde{\type}_i;G_{\allocspace},\hat{s}_i)$. Because the required signal is independent of the reports of other agents, the signaling choice of type $\hat{\type}_i$ is uncoordinated. We want to show (1) \textbf{utility improvement}: the candidate mechanism improves interim utilities by removing coordination, and (2) \textbf{incentive compatibility}: the construction is valid in the sense that the candidate mechanism is incentive-compatible.

\textbf{Utility improvement.} To show (1), we only need to show that the interim utility of type $\hat{\type}_i$ increases, since the rest of the mechanism remains the same as the original one. Thus, it suffices to show that the required signal $\hat{s}_i$ is weakly lower than type $\hat{\type}_i$'s certainty-equivalent signal under the original mechanism, $\ce_i(G,\hat{\type}_i)$.

To see this, consider an intermediary mechanism $\widehat{\mathcal{M}}$ in which type $\hat{\type}_i$ is recommended to produce a deterministic signal $\ce_i(G,\hat{\type}_i)$ and receives the item following the marginal distribution $G_{\allocspace}$, leaving the rest of the mechanism the same as the original mechanism $\mathcal{M}$. By construction, the interim utility, the expected utility of every type from truthful reporting, is the same as that in the original mechanism $\mathcal{M}$.\footnote{We highlight that the main purpose of introducing the intermediary mechanism $\widehat{\mathcal{M}}$ is to show (1). The incentive compatibility issue of the intermediary mechanism $\widehat{\mathcal{M}}$ does not affect our analysis, and so we will skip its discussion.} Hence, we also denote the interim utility in $\widehat{\mathcal{M}}$ by $\util_i(\type_i)$.
Moreover, the intermediary mechanism $\widehat{\mathcal{M}}$ differs from our candidate mechanism $\mathcal{M}\primed$ only in the deterministic signal recommended for type $\hat{\type}_i$. 

In addition, we introduce the deviation utilities, each type $\type_i$'s expected utility from mimicking $\hat{\type}_i$, in all three mechanisms that will be helpful for the analysis: (a) in the original mechanism $\mathcal{M}$: $\util_i(\type_i;G)$, which equals $\util_i(\type_i;G_{\allocspace},\ce_i(G,\type_i))$, (b) in the candidate mechanism $\mathcal{M}\primed$: $\util_i\primed(\type_i;G_{\allocspace},\hat{\signal}_i)$, and (c) in the intermediary mechanism $\widehat{\mathcal{M}}$: $\util_i(\type_i;G_{\allocspace},\ce_i(G,\hat{\type}_i))$. 
These three deviation utilities correspond to the red, green, and blue curves in \Cref{fig:intuition_coordination}, respectively. In addition, the black curve represents $\util_i(\type_i)$, the interim utility in $\mathcal{M}$ and $\widehat{\mathcal{M}}$. 

\begin{figure}[t]
\centering
\begin{tikzpicture}[xscale=5,yscale=3.5, dot/.style={draw,circle,fill,minimum size=0.6mm,inner sep=0pt},
    pins/.style={#1, pin edge={<-, #1, decorate, decoration={name=lineto,
    pre=moveto, pre length=2pt}}},
    Dotted/.style={
    dash pattern=on 0.1\pgflinewidth off #1\pgflinewidth,line cap=round,
    shorten >=#1\pgflinewidth/2,shorten <=#1\pgflinewidth/2},
    Dotted/.default=4]

\draw [<->] (0,1.1) -- (0,0) -- (1.1,0);
\node [below] at (1.15, 0 ) {$\type_i$};
\node [left] at (0,1.1 ) {$\util_i$};
\draw [blue, thick] plot [smooth, tension=0.6] coordinates { (0,0) (0.3, 0.16) (0.5,0.5) };
\draw [blue, thick] plot [smooth, tension=0.6] coordinates { (0.5,0.5) (0.7,0.7) (1, 0.8)};
\node [right, black] at (1, 0.8 ) {$\util_i(\theta_i;\textcolor{blue}{G_{\allocspace},\ce_i(G,\hat{\type}_i)})$};
\node [right, black] at (1, 0.95 ) {$\util_i(\theta_i;\textcolor{red}{G})$};

\draw [red, thick] plot [smooth, tension=0.6] coordinates { (0,0.1) (0.3, 0.33) (0.5,0.5) (0.7,0.7) (0.8, 0.85) (1,0.95)};
\node [fill, draw, circle, minimum width=3pt, inner sep=0pt, 
      pin={[fill=white, outer sep=1pt]315:same util}] 
      at (.5,0.5) {};

\draw [black, thick] plot [smooth, tension=0.6] coordinates { (0,0.33) (0.3, 0.4) (0.5,0.5) (0.7,0.8) (1, 1.1)};
\node [right, black] at (1, 1.1) {$\util_i(\theta_i)$};

\draw [green!60!black, thick] plot [smooth, tension=0.6] coordinates { (0,0.2) (0.3, 0.4) (0.5,0.7) (0.7,1) (1, 1.2)};
\node [right, black] at (1, 1.28 ) {$\util_i\primed(\theta_i;\textcolor{green!60!black}{G_{\allocspace},\hat{\signal}_i})$};

\node [] at (0.5, -0.1) {$\hat{\type}_i$};
\draw [dotted] (0.5,0) -- (0.5,0.5);
\node (A) at (0,0) {};
\node (B) at (0.5,0) {};

\node [] at (0.3, -0.1) {$\tilde{\type}_i$};
\draw [dotted] (0.3,0) -- (0.3,0.4);

\end{tikzpicture}
\caption{Comparison of Incentive Cost: Coordination vs Non-coordination}
\rule{0in}{1.2em}$^\dag$\scriptsize 
The black curve is the interim utility of the agents in the original mechanism $\mathcal{M}$. 
The \textcolor{red}{red curve} captures the deviation utility for each type by mimicking  type $\hat{\type}_i$ in the original coordination mechanism $\mathcal{M}$. The \textcolor{blue}{blue curve} captures the deviation utility for each type by mimicking type $\hat{\type}_i$ in the intermediary mechanism $\widehat{\mathcal{M}}$. The \textcolor{green!60!black}{green curve} captures the deviation utility for each type by mimicking type $\hat{\type}_i$ in the candidate mechanism $\mathcal{M}\primed$.\\
\label{fig:intuition_coordination}
\end{figure}

We argue that by replacing the stochastic signal recommendation in the original mechanism $\mathcal{M}$ with a deterministic one in the intermediary mechanism $\widehat{\mathcal{M}}$, we relax the upward deviation incentive of the lower type~$\tilde{\type}_i$. 
This is implied by combining the observations that (i) the interim utility of type $\tilde{\type}_i$ remains unchanged in all three mechanisms, and (ii) type $\tilde{\type}_i$'s deviation utility from mimicking type $\hat{\type}_i$ decreases in the intermediary mechanism $\widehat{\mathcal{M}}$. 
The latter claim is true because the  ICE condition implies that $\ce_i(G,\type_i)\le \ce_i(G,\hat{\type}_i)$ for all $\type_i\le \hat{\type}_i$. Visually, this means that the red curve lies above the blue curve for all types below $\hat{\type}_i$.\footnote{Careful readers may wonder why there are multiple crossings between the blue and red curves. Recall that our single-crossing condition only requires comparisons across deterministic outcomes (\cref{asp:weak_single_crossing}), whereas a coordination mechanism could offer stochastic recommendations. Thus, the red and blue curves may exhibit multiple crossings \citep{quah2012aggregating}.} 
Combining (i) and (ii), we conclude that the upward-deviation incentive of type $\tilde{\type}_i$ is relaxed in $\widehat{\mathcal{M}}$ relative to $\mathcal{M}$. 
Intuitively, this slackness allows us to further reduce the deterministic signal recommended to type $\hat{\type}_i$.

To make this intuition concrete, compare the candidate mechanism $\mathcal{M}\primed$ with the intermediary mechanism $\widehat{\mathcal{M}}$.
Recall that the candidate mechanism $\mathcal{M}\primed$ is constructed such that type $\tilde{\type}_i$'s interim utility is the same as the deviation utility from mimicking type $ \hat{\type}_i$. This means the green and black curves intersect at type $\tilde{\type}_i$.
Hence, incentive compatibility in the original mechanism implies that at type $\tilde{\type}_i$, the green (black) curve is above the red curve, thus also above the blue curve.
Given that each agent's utility decreases in the signal (\cref{asp:monotonicity}), the green curve above the blue at type $\tilde{\type}_i$ implies that $\hat{s}_i\le \ce_i(G,\hat{\type}_i)$. 
Thus, we establish that our candidate mechanism improves utility for $\hat{\type}_i$ by removing coordination. 

Note that when strict ICE holds, $\ce_i(G,\type_i)< \ce_i(G,\hat{\type}_i)$ for all $\type_i< \hat{\type}_i$ if the signal recommendation in $G$ is non-degenerate. This implies that the red curve is strictly above the blue curve at type $\tilde{\type}_i$. 
Intuitively, this means that by replacing a stochastic recommendation with a deterministic one, the upward deviation incentive for the lower type  $\tilde{\type}_i$ becomes strictly lower in the deterministic mechanism.
This strict slackness in type  $\tilde{\type}_i$'s upward deviation incentive enables us to strictly reduce the deterministic signal recommended to type $\hat{\type}_i$.
In \cref{fig:intuition_coordination}, this means that the green curve is also strictly above the blue curve at type $\tilde{\type}_i$, and hence $\hat{s}_i< \ce_i(G,\hat{\type}_i)$. To summarize, under the strict ICE, the utility improvement from replacing a stochastic recommendation with a deterministic one is also strict. In other words, coordination mechanisms are strictly sub-optimal under the strict ICE.  

\textbf{Incentive compatibility.}
Now we argue that (2) holds. Notice that we construct our candidate mechanism so that the local upward-deviation constraint for type $\tilde{\type}_i$ is binding.
Then, given a monotone allocation rule, the single-crossing condition implies that the local downward-deviation constraint for type $\hat{\type}_i$ also holds.\footnote{Moreover, by construction, both the local upward-deviation constraint and the local downward-deviation constraint for comparisons among other types remain valid in the candidate mechanism.} 
Recall that our single-crossing condition only requires comparisons across deterministic outcomes (\cref{asp:weak_single_crossing}); we cannot apply it here to rule out global deviation because our current candidate mechanism may still involve stochastic recommendations for types above $\hat{\type}_i$.
In the formal proof, we circumvent this by modifying the candidate mechanism as follows: construct our candidate mechanism inductively so that the signal recommendations are deterministic for all types, and all local upward-deviation constraints are binding.
As a result, the single-crossing condition ensures that all the local downward-deviation constraints also hold.
Moreover, the single-crossing property also ensures that no global deviation is profitable.
\bigskip

In our construction, upward-deviation incentives are the primary concern, in contrast to classical revenue-maximizing auction design where downward-deviation incentives typically bind. This distinction stems from the differing design objectives.
In revenue maximization, the first-best benchmark (absent incentive constraints) involves charging each agent their full valuation. Private information then introduces downward-deviation incentives, as higher types mimic lower types to retain informational surplus. Conversely, because our objective is to maximize agent utilities, the first-best benchmark for any monotone interim allocation rule would set the signaling cost to zero. Consequently, the economic tension reverses: private information creates upward-deviation incentives, as now lower types want to mimic higher types to secure a larger allocation without the deterrent of high costs. This comparison illustrates why upward-deviation constraints bind in utility-maximization environments, while downward-deviation constraints bind in revenue-maximization settings.

\paragraph{Payoff Non-equivalence.}
Theorem~\ref{thm:monotone optimal general} operates in an environment where payoff equivalence fails: fixing an implementable (interim) allocation rule does not pin down agents' interim utilities, even after normalizing the lowest type's utility. The main driver for this non-equivalence is the type-dependent signaling cost. As we have seen in our introductory example, different implementations of the same allocation rule can induce different distributions of signal requirements across types. In environments with quasi-linear payments or type-independent signaling costs, different distributions of signal requirements have the same certainty-equivalent signal. However, under type-dependent signaling costs, different distributions of signal requirements result in different certainty-equivalent signals, which eventually translate directly into different information rents. 

It is also useful to contrast this source of non-equivalence with the one in signaling games. In signaling games, there is usually one sender (the agent) and one receiver (the principal).
There, the same allocation outcome can result in different signaling costs because the principal could have different off-path beliefs. 
Refinements such as D1 address precisely this belief-driven indeterminacy by selecting the Riley outcome (minimal-cost separating equilibrium). 
In our model, by contrast, payoff non-equivalence is not driven by belief selection, but reflects real differences in the signaling costs required to implement a given allocation rule. 
To see this, since there is only one agent in the Spencian signaling game, there is no scope for coordination by assumption: the agent chooses their own signal based solely on their own private type\footnote{An alternative interpretation is that there is a continuum of agents. Regardless of the interpretation, each agent chooses their own signal based solely on their own private type in the Spencian signaling game.}. 
In contrast, in our setting, there is a non-trivial tradeoff in how the designer uses cross-agent information to coordinate signal requirements. 

\paragraph{Optimal Non-coordination Mechanism.}
The optimality of non-coordination substantially reduces the complexity of the design space. By ruling out mechanisms that rely on cross-agent coordination in signal recommendations, it narrows attention to a tractable class of implementations and provides a practical path toward characterization and implementation.
For example, in a follow-up work, \citet{li2026} study resource allocation when agents can misrepresent their types and compete with one another. In their setting, each agent has utility $x_i-(\signal_i-\type_i)^+$, which satisfies the ICE assumption. As a result, by invoking \cref{thm:monotone optimal general}, they can, without loss of optimality, restrict attention to non-coordination mechanisms, which enables a characterization of the optimal design. 
\citet{li2026} further apply this characterization in large markets and show that the optimal mechanism may distort the efficient allocation, either at the middle or at the top, in order to maximize agents' utilities.

\section{Indirect Implementation: Coarse Ranking Contests}
\label{sub:implementation}
There is a tight connection between  non-coordination mechanisms and contests.
In this section, we assume that the environment is symmetric, that is, all agents have the same utility function and follow the same type distribution. 
We will show that in the symmetric environment, coarse ranking contests can implement symmetric direct non-coordination mechanisms, where the allocation rule for each agent is the same. 
Specifically, we will propose the concept of coarse ranking, in which segments of signals are pooled and assigned the same rank. Correspondingly, we extend the usual notion of a contest to the so-called coarse ranking contest, where the items are allocated to the $k$ agents with the highest coarse rankings, with ties broken uniformly at random. 

\begin{definition}[coarse ranking]
Given any countable set of disjoint open intervals $\{(\underline{\signal}^{(j)}, \bar{\signal}^{(j)})\}_{j=1}^{\infty}$ whose union is a subset of the signal space, 
the \emph{coarse ranking} of agent $i$ under the signal profile $\signals=(\signal_1,\dots,\signal_n)$ is 
$$\ranking_i(\signals) = \abs{\lbr{i'\neq i, 1\leq i'\leq n: \signal_{i'} > \bar{\signal}^{(j)}}},$$ 
and the number of ties for agent $i$ is
$$\tie_i(\signals) = \abs{\lbr{i'\neq i, 1\leq i'\leq n: 
\bar{\signal}^{(j')} = \bar{\signal}^{(j)}}}+1.$$
Here, for any signal $\signal_i$, $j$ is the index such that $\signal_i \in (\underline{\signal}^{(j)}, \bar{\signal}^{(j)})$, if such a $j$ exists (i.e., if $\signal_i$ falls into one of the intervals defined). If no such $j$ exists (i.e., if $\signal_i$ lies outside all the intervals defined), then (slightly overloading the notation) we let $\bar{\signal}^{(j)} = \underline{\signal}^{(j)} = \signal_i$.
We also call the pair of functions $(\mathbf\ranking,\mathbf\tie)$ a \emph{coarse ranking}.
\end{definition}

Intuitively, each interval in the set 
$\{(\underline{\signal}^{(j)}, \bar{\signal}^{(j)})\}_{j=1}^{\infty}$  specifies a region of signals that are pooled and assigned the same (coarse) ranking.
Outside the closure of the union of these intervals, signals are ranked strictly.
In the special case where $\{(\underline{\signal}^{(j)}, \bar{\signal}^{(j)})\}_{j=1}^{\infty}$ is empty, the definition of a coarse ranking coincides with the usual definition of a strict ranking, where the agents' ranks are given by the order of their signals in the given signal profile.

Our generalization of strict rankings to coarse rankings leads to a larger class of contest rules, defined as follows. For any coarse ranking $(\mathbf\ranking,\mathbf\tie)$ and
any agent $i$, the induced (coarse ranking) contest rule is
\begin{align*}
\tilde\signalalloc_i(\signals;\mathbf\ranking,\mathbf\tie ) = 
\begin{cases}
1, & k\geq \ranking_i(\signals) + \tie_i(\signals), \\
\frac{k-\ranking_i(\signals)}{\tie_i(\signals)}, 
& k\in (\ranking_i(\signals), \ranking_i(\signals) + \tie_i(\signals)), \\
0, & k\leq \ranking_i(\signals).
\end{cases}
\end{align*}

The rule can be understood in two steps. First, each nondegenerate interval forms a rank group: all agents whose signals lie in that interval are treated as tied, while signals outside the pooling intervals retain their strict ordering. Thus, $\ranking_i(\signals)$ is the number of agents ranked strictly above agent $i$, and $\tie_i(\signals)$ is the size of agent $i$'s rank group, including agent~$i$. Second, the $k$ items are allocated from the top rank downward. If all members of agent $i$'s rank group fall within the top $k$ positions, agent $i$ receives an item; if the group lies entirely below the cutoff, agent $i$ receives no item; and if the group straddles the cutoff, the $k-\ranking_i(\signals)$ remaining items are allocated uniformly among the $\tie_i(\signals)$ tied agents. In particular, when there is one item, an agent who alone occupies the highest rank group wins for sure; if several agents occupy the highest rank group, they share the item uniformly.

\begin{definition}[coarse ranking contest rule]\label{def: coarse ranking contest rule}
A mapping profile $\signalallocs:\Signalspace\rightarrow X$ is a \emph{coarse ranking contest rule} if there exists a coarse ranking $(\mathbf\ranking,\mathbf\tie)$ such that $\signalalloc_i(\signals)=\tilde\signalalloc_i(\signals;\mathbf\ranking,\mathbf\tie )$ for each $i$.
\end{definition}

The class of coarse ranking contest rules, which is a subset of the class of all mappings from the signal space to the allocation space, also includes the class of contest rules that allocate items (prizes) to agents based on the strict ranking of their signals such as all-pay contests \citep{baye1996all}.
The proof of the proposition is provided in \cref{apx:discussion}.

\begin{proposition}\label{prop:implement_via_majorization} 
In symmetric environments, any symmetric interim allocation--utility pair $(\alloc,\util)$ 
that can be implemented by a direct non-coordination mechanism
has an indirect implementation that is a randomization over the coarse ranking contests.
\end{proposition}

\section{Discussions}
\label{sec:discussions}
In this section, we establish a partial converse to our main result, showing that coordination mechanisms can outperform non-coordination mechanisms when agents' utility functions satisfy a strictly decreasing certainty equivalence condition.
We also show that in our model, it is possible for the principal to implement a non-monotone allocation, and we identify sufficient conditions under which allowing non-monotone allocations does not alter our main conclusions.

\subsection{Sub-optimality of Non-coordination}
We now present a partial converse to \cref{thm:monotone optimal general}, showing that if Assumption \ref{asp:monotone_certainty_equivalence} is violated, then direct mechanisms can strictly outperform non-coordination mechanisms.
Specifically, when agents' certainty equivalence is strictly decreasing in type, coordination through stochastic signal recommendations can yield strictly higher utility for at least one type, even when the underlying allocation rule remains unchanged.

\begin{assumption}[Strictly decreasing certainty equivalence]\label{asp:decreasing_certainty_equivalence}
For any non-degenerate distribution $G$ over allocations and signals and for any agent $i$, his certainty equivalent signal $\ce_i(G,\type)$ is strictly decreasing in $\type$. 
\end{assumption}

\begin{proposition}[Sub-optimality of non-coordination]
\label{prop:monotone optimal general}
Suppose Assumptions~\ref{asp:monotonicity}, \ref{asp:weak_single_crossing}, and \ref{asp:decreasing_certainty_equivalence} hold.
Fix any $(\allocs,\utils)$  that can be implemented by a non-coordination mechanism and for which $\allocs$ is strictly increasing.
Then there exist a distribution $\boldsymbol{H}$ and a pair $(\allocs\primed,\utils\primed)$ that can be implemented by a direct mechanism with stochastic  signal recommendation. Moreover, $\allocs\primed=\allocs$ and 
$\util\primed_i(\type_i) \geq \util_i(\type_i)$ for any $i$ and type $\type_i$, where the inequality is strict for at least one type. 
\end{proposition}
The proof is provided in \cref{apx:suboptimality}.
Note that \cref{prop:monotone optimal general} only implies that non-coordination mechanisms can be improved upon by coordination mechanisms with stochastic signal recommendation. 
The result is silent on which coordination structure is optimal for the implementation. 

\subsection{The Implementability of Non-Monotone Allocations}
We focus on implementing monotone allocations, motivated by both practical and ethical fairness considerations \citep[see, e.g.,][]{roth2002economist,gershkov2022optimal}. Monotonicity offers a transparent and easily interpretable structure that aligns with widely held normative intuitions: individuals who demonstrate greater merit, exert more effort, or exhibit greater need should not be disadvantaged relative to others. This principle resonates in domains such as college admissions, where stronger academic records are expected to improve admission chances, and in means-tested programs, where lower incomes should yield more generous benefits \citep{budish2011combinatorial}. In such contexts, monotonicity not only enhances predictability but also reinforces institutional legitimacy by aligning with public expectations and reducing the risk of perceived arbitrariness or political backlash.

In contrast, non-monotone rules are harder to justify in public-facing mechanisms where transparency and accountability are critical. They may conflict with intuitive fairness norms and lead stakeholders to view the process as opaque or capricious. When improved qualifications yield worse outcomes, this can generate confusion and resentment and erode trust, even when such rules are defensible within a formal model.

Nonetheless, from a theoretical standpoint, it is interesting to explore the boundaries of implementability and consider when non-monotone allocation rules can arise. The following example demonstrates that non-monotone allocations may be implementable.

\begin{example}\label{exp:non_monotone}
Consider a simplified single-agent screening problem in which the agent's type $\type$ is uniformly distributed on $[0,1]$ and his utility is $\signalalloc-(\signal-\type)^+$. We describe the direct mechanism and its timing explicitly. After observing his type, the agent reports $\hat{\type}\in[0,1]$. Reports in
$[0,\frac{1}{2}]\cup[\frac{3}{4},1]$ receive the randomized option $R$, while reports in $(\frac{1}{2},\frac{3}{4})$ receive the deterministic option $D$. The options are as follows:
\begin{itemize}
    \item Under $D$, the mechanism recommends signal $\signal=\frac{2}{3}$ and promises allocation probability $\signalalloc=\frac{1}{2}$.
    \item Under $R$, with probability $\frac{2}{3}$ the mechanism recommends $\signal=1$ and promises $\signalalloc=1$, and with probability $\frac{1}{3}$ it recommends $\signal=0$ and promises $\signalalloc=0$.
\end{itemize}
After observing the recommendation, the agent chooses a signal. The promised allocation is delivered if he follows the recommendation; otherwise, his allocation is zero.

The utility from option $R$ is
\[
U_R(\type)=\frac{2}{3}\left(1-(1-\type)^+\right)=\frac{2}{3}\type.
\]
The value of option $D$, including the option to opt out after seeing its recommendation, is
\[
U_D(\type)=\max\left\{\frac{1}{2}-\left(\frac{2}{3}-\type\right)^+,0\right\}
=
\begin{cases}
0, & 0\leq\type\leq\frac{1}{6},\\
\type-\frac{1}{6}, & \frac{1}{6}<\type<\frac{2}{3},\\
\frac{1}{2}, & \frac{2}{3}\leq\type\leq1.
\end{cases}
\]
Hence $R$ is strictly preferred for $\type\in(0,\frac{1}{2})\cup(\frac{3}{4},1]$, whereas $D$ is strictly preferred for $\type\in(\frac{1}{2},\frac{3}{4})$. The agent is indifferent at the boundary types and thus assigning the boundary reports to $R$ gives a truthful mechanism. The induced interim allocation is
\[
\alloc(\type)=
\begin{cases}
\frac{2}{3}, & \type\in[0,\frac{1}{2}]\cup[\frac{3}{4},1],\\
\frac{1}{2}, & \type\in(\frac{1}{2},\frac{3}{4}).
\end{cases}
\]
This allocation is non-monotone in the natural order of types: it falls from $\frac{2}{3}$ to $\frac{1}{2}$ as the type crosses $\frac{1}{2}$, and rises from $\frac{1}{2}$ to $\frac{2}{3}$ as the type crosses $\frac{3}{4}$.
\end{example}

We can extend this example to a multi-agent setting, showing that when the principal leverages cross-agent information to generate signal recommendations, she can implement a broader class of allocation rules, including some that are non-monotone. However, implementing non-monotone allocations may be suboptimal from the principal's perspective. In \cref{sub: general allocations}, we demonstrate that under an additional assumption, any implementable non-monotone allocation rule is dominated by a monotone allocation rule, in the sense that implementing the latter via a non-coordination mechanism yields weakly higher utility for all types. Furthermore, when the principal favors assortative matching that assigns higher allocations to higher types, implementing the monotone allocation rule using a non-coordination mechanism strictly improves the principal's objective.

\subsection{Optimality of Monotone Allocations}\label{sub: general allocations}
As illustrated in \cref{exp:non_monotone}, certain non-monotone interim allocations can be implemented via direct mechanisms.
This arises when agents' utility functions fail to satisfy the single-crossing property under stochastic (lottery-based) signal recommendations.

We introduce ``no concave crossing'', an assumption that regulates the agent's preferences between deterministic and stochastic signal recommendations.
Under this assumption, we can show that any non-monotone allocation rule that can be implemented by a direct mechanism leaves each agent with a lower utility.
In other words, we can always find another monotone allocation rule such that it can be implemented by a direct non-coordination mechanism that leaves each agent with a higher utility.
This implies that focusing on monotone allocations in \cref{thm:monotone optimal general} is without loss of optimality. 

\begin{assumption}[no-concave-crossing]
\label{asp:no_concave_crossing}
For any agent $i$, consider any allocation and deterministic signal recommendation $(\signalalloc_i,\signal_i)$, and any allocation and stochastic signal recommendation $(\signalalloc'_i,\dista_i)$. 
If there exist two types $\hat{\type}_i<\hat{\type}'_i$ such that
\begin{align*}
\expostutil_i(\signalalloc_i,\signal_i,\hat{\type}_i) &\geq \expect[\signal\sim\dista_i]{\expostutil_i(\signalalloc'_i,\signal,\hat{\type}_i)}\\
\expostutil_i(\signalalloc_i,\signal_i,\hat{\type}'_i)  &\geq \expect[\signal\sim\dista_i]{\expostutil_i(\signalalloc'_i,\signal,\hat{\type}'_i)},
\end{align*}
then for any $\type_i \in [\hat{\type}_i, \hat{\type}'_i]$, we have the following:
$\expostutil_i(\signalalloc_i,\signal_i,\type_i)  \geq \expect[\signal\sim\dista_i]{\expostutil_i(\signalalloc'_i,\signal,\type_i)}$.
\end{assumption}

Note that the above condition compares any deterministic outcome with any lottery. 
It implies that if two types weakly prefer the deterministic outcome to the lottery, then every intermediate type must also weakly prefer the deterministic outcome.
Our condition is weaker than the one studied in \citet{kartik2024single}, which considers two arbitrary lotteries. 
Examples satisfying our ``no concave crossing'' condition include $\expostutil_i(\signalalloc_i,\signal_i,\type_i) = v_i(\signalalloc_i,\type_i)- c_i(\signal_i,\type_i)$ when $c_i(\signal_i,\type_i) = (\signal_i-\type_i)^+$ 
and  $c_i(\signal_i,\type_i) = (\signal_i-\type_i)^2$. 

\begin{theorem}\label{thm:convex_monotone}
Under Assumptions \ref{asp:monotonicity}, \ref{asp:weak_single_crossing}, \ref{asp:monotone_certainty_equivalence} and \ref{asp:no_concave_crossing}, 
for any interim allocation–utility pair $(\allocs,\utils)$ 
that can be implemented by a direct mechanism, 
there exists another interim allocation–utility pair $(\allocs\primed,\utils\primed)$ such that (1) it can be implemented by a non-coordination mechanism; (2)  $\allocs\primed$ is a monotone allocation rule; (3) $\util_i\primed(\type_i)\geq \util_i(\type_i)$ for all $i$.
\end{theorem}

The proof is provided in \cref{appx B: general allocations}.
The key idea of the proof is to construct a direct non-coordination mechanism that implements the monotone rearrangement of a given (possibly non-monotone) allocation rule, denoted  $\alloc_i\primed$. We then show that this rearranged allocation leads to weakly higher utility for all types and strictly higher utility for at least one type.
Formally, given any non-monotone allocation rule $\allocs$, define its monotone rearrangement $\allocs^\dagger$ as follows. For each $i$, let 
$
G_i(z) = \int_{\underline{\theta}_i}^{\overline{\theta}_i} \mathbf{1}\{ \alloc_i(\theta) \leq z \} \, dF_i(\theta)
$ be the distribution function of the original allocation rule.
The \textbf{monotone rearrangement} (or quantile transform) of \( \alloc_i \) is defined by:

\[
\alloc_i^\dagger(\theta) = G_i^{-1}(F_i(\theta)),
\]
where \( G_i^{-1}(u) = \inf\{ z \in \mathbb{R} : G_i(z) \geq u \} \) is the generalized inverse of the distribution function \( G_i \).

Recall that the principal's objective is to maximize a weighted average of agents' utilities, possibly combined with other goals such as matching efficiency.
To compare alternative allocation rules, we consider a reduced-form objective that includes both the total matching efficiency
$$\sum_{i=1}^n \expect[\type_i\sim\dist_i]{W_i(\alloc_i(\type_i),\type_i)}$$
and the weighted average of agents' utilities. If the principal's objective favors assortative matching, that is, the total efficiency is higher under the monotone rearrangement $\allocs^\dagger$, then implementing $\allocs^\dagger$ via a direct non-coordination mechanism strictly improves upon implementing the original non-monotone allocation rule.

\section{Conclusions}
\label{sec:conclusion}
In this paper, we study the design of optimal screening mechanisms for allocating limited resources among multiple agents based on costly signals. We show that, to minimize signaling costs, it is optimal not to coordinate agents' signaling choices using others' reports. That is, non-coordination mechanisms, where each agent's recommendation depends only on their own type, maximize agents' utilities while implementing a given allocation rule.

This finding not only identifies a novel welfare gain in screening environments with endogenous costs but also has broader implications for mechanism design. In particular, the optimality of non-coordination significantly reduces the complexity of the design space. By ruling out mechanisms that rely on cross-agent coordination, it simplifies the search for optimal solutions and offers a tractable path forward for practical implementation.

These insights open several promising directions for future research. One is to explore the structure of optimal mechanisms under more general or heterogeneous cost functions, including those that capture non-standard preferences or domain-specific constraints. Another is to study how non-coordination principles extend to dynamic or multi-stage environments, or interact with equity and redistribution goals. Such investigations may yield deeper theoretical characterizations and inform the design of robust, low-cost allocation mechanisms in education, procurement, housing, and beyond.

\appendix

\section{Proof of \cref{thm:monotone optimal general}}
\label{apx:main thm}
To simplify the exposition, we assume that the type space is finite.
We denote the type space by $\Theta_i=\{\hat{\type}^{(0)}_i,\dots , \hat{\type}^{(m)}_i\}$, with $\hat{\type}^{(0)}_i < \dots < \hat{\type}^{(m)}_i$. 
The extension to continuous types can be shown via standard discretization tricks, with details provided in \cref{subapx:continuous}.

Fixing any direct mechanism that delivers the interim allocation $\allocs$ and the interim utility~$\utils$, let $G_i^{(k)}$ be the distribution over the allocations and signals recommended to the agent $i$ who reports as type $\hat{\type}^{(k)}_i$.
Let $\alloc_i^{(k)}\in[0,1]$ be the probability of  allocating the item to agent $i$ induced by $G_i^{(k)}$.
Denote the certainty equivalent signal of $G_i^{(k)}$ for type $\hat{\type}^{(k)}_i$ by 
\begin{align*}
\signal_i^{(k)} \triangleq \ce(G_i^{(k)},\hat{\type}^{(k)}_i), \quad 0\leq k\leq m.
\end{align*}

As illustrated in the main text, the key proof idea is to construct a mechanism that prescribes deterministic signal recommendations for agents and show that this constructed mechanism is implementable by a non-coordination mechanism and improves agents' utilities. 
The key idea for the construction is to find a deterministic signal recommendation that prevents local upward deviation.

Denote the signal recommended to type $\hat{\type}^{(k)}_i $ by
$\hat\signal_i^{(k)}$, for all $0\leq k\leq m$. 
We proceed inductively.
Let the signal recommended to type $\hat{\type}^{(0)}_i $ be
$\hat\signal_i^{(0)} := \signal_i^{(0)}$.
For all $k\in \{1,\dots,m\}$  and all $i\in \{1,\dots,n\}$, let $\hat{\signal}_i^{(k)}$, the signal recommended to type $\hat{\type}^{(k)}_i $,  be the solution of
\begin{align*}\label{eq:construction}\tag{construction}
\util_i\rbr{\hat{\type}^{(k-1)}_i;\alloc^{(k)}_i,\hat{\signal}_i^{(k)}} = 
\util_i\rbr{\hat{\type}^{(k-1)}_i;\alloc^{(k-1)}_i,\hat{\signal}_i^{(k-1)}}. 
\end{align*}

Let $\util_i\primed(\hat{\type}^{(k)}_i):=\util_i\rbr{\hat{\type}^{(k)}_i;\alloc^{(k)}_i,\hat{\signal}_i^{(k)}}$.
By construction, we have $\util_i\primed\rbr{\hat{\type}^{(0)}_i}=\util_i\rbr{\hat{\type}^{(0)}_i}$.

It remains to show that (1) $(\allocs,\utils\primed)$ can be implemented as a direct non-coordination mechanism (\textbf{Step 1}); and (2) $\util_i\primed(\hat{\type}^{(k)}_i)\geq\util_i(\hat{\type}^{(k)}_i)$ for all $i$ and $k$ (\textbf{Step 2}). 

\paragraph{Step 1:}
Recall that, as defined in \cref{def:implementability}, to show implementability, we need to show consistency and incentive compatibility. 
First, we show consistency.
Given that $\allocs$ is an interim allocation, there exists some ex post allocation rule $\mechallocs:\Typespace\to \allocspace$ that induces $\allocs$.
If the signal recommendation $\signalrecommend_i(\type_i)$ is an injection, we can perfectly infer the type from the recommended signals, and by letting $x_i(\signals) := \mechalloc_i(\type_1(\signal_1), \dots, \type_n(\signal_n))$ where $\type_i(\signal_i)$ is the type whose recommended signal is $\signal_i$, 
we implement the interim allocation $\allocs$ under the non-coordination mechanism. 
If there are multiple types that are recommended the same signal, by the incentive compatibility constraints, their interim allocations must be the same. 
Therefore, we can simulate $\type_i(\signal_i)$ according to the prior distribution conditional on the recommended signal being $\signal_i$, and let $x_i(\signals) := \mechalloc_i(\type_1(\signal_1), \dots, \type_n(\signal_n))$.
This again ensures the implementation of the interim allocation $\allocs$.
Second, we show that incentive compatibility holds. By construction, there is no incentive for local upward deviation. Given that each agent's utility is weakly increasing in type (\cref{asp:monotonicity}) and the fact that local upward deviation incentives are binding, the weak single crossing property (\cref{asp:weak_single_crossing}) implies that there is no incentive for local downward deviation.
Applying the single crossing property again, local incentive compatibility ensures global incentive compatibility.
Therefore, $(\allocs,\utils\primed)$ can be implemented as a direct non-coordination mechanism.

\paragraph{Step 2:} 
To show improvement in utility, the key is to show that the deterministic signal recommendation for each type in the constructed mechanism is lower than the certainty equivalent signal given the stochastic recommendation in the original mechanism.
We prove by induction. Fix any $i$ and suppose the following induction hypothesis holds.
\begin{align*}\label{eq:induction}\tag{induction}
\util_i\primed\rbr{\hat{\type}^{(k-1)}_i}\geq\util_i\rbr{\hat{\type}^{(k-1)}_i}
\end{align*}

\begin{claim}\label{claim}
    Since $(\allocs,\utils)$ can be implemented as a direct  mechanism, the following incentive compatibility condition must hold:
\begin{align*}\label{eq:incentive}\tag{incentive}
\util_i\rbr{\hat{\type}^{(k-1)}_i}\geq \util_i\rbr{\hat{\type}^{(k-1)}_i;\alloc^{(k)}_i,\signal_i^{(k)}}.
\end{align*}
\end{claim}

Let us take a detour and understand the remaining step intuitively using \cref{fig:intuition_coordination}.
Suppose $\hat{\type}^{(k-1)}_i$ is the type $\tilde{\type}_i$ and $\hat{\type}^{(k)}_i$ is the type $\hat{\type}_i$ in \cref{fig:intuition_coordination}.
Intuitively, \eqref{eq:incentive} guarantees that the red curve is above the blue curve at $\tilde{\type}_i$.
Notice that \eqref{eq:construction} ensures that each type is indifferent between local upward deviation and truthful reporting in the constructed candidate mechanism. 
Intuitively, \eqref{eq:construction} means that the value of the green curve at $\tilde{\type}_i$ equals type $\tilde{\type}_i$'s expected utility from truthful reporting.
Together with \eqref{eq:induction}, we can infer that the green curve is above the red curve at type $\tilde{\type}_i$.
Hence, the green curve is above the blue curve at type $\tilde{\type}_i$.
This implies that
$\hat{\signal}^{(k)}_i\leq \signal^{(k)}_i$.

Formally, combining \eqref{eq:construction}, \eqref{eq:induction}, and \eqref{eq:incentive}, we have
\begin{align*}
\util_i\rbr{\hat{\type}^{(k-1)}_i;\alloc^{(k)}_i,\hat{\signal}_i^{(k)}}\geq \util_i\rbr{\hat{\type}^{(k-1)}_i;\alloc^{(k)}_i,\signal_i^{(k)}}.
\end{align*}
This implies that
$\hat{\signal}^{(k)}_i\leq \signal^{(k)}_i$. In other words, the deterministic signal recommendation for each type in the constructed mechanism is lower than the certainty equivalent signal.
As a result, this improves each type's utility in the constructed mechanism:  $\util_i\primed(\hat{\type}^{(k)}_i)\geq\util_i(\hat{\type}^{(k)}_i)$.

Induction implies that $\util_i\primed(\hat{\type}^{(k)}_i)\geq\util_i(\hat{\type}^{(k)}_i)$ for all $i$ and $k$. 
Combining this with Step 1 establishes \cref{thm:monotone optimal general}. 

\begin{proof}[Proof of \cref{claim}]
We use $H_i^{(k)}$ to denote
the marginal distribution over allocation given to agent $i$ who reports type $\hat{\type}^{(k)}_i$.
Notice that the mean of $H_i^{(k)}$ is $\alloc_i^{(k)}$.
Denote the certainty equivalent signal of $G_i^{(k)}$ for type $\hat{\type}^{(j)}_i$ by 
\begin{align*}
\signal_i^{(k,j)} \triangleq \ce(G_i^{(k)},\hat{\type}^{(j)}_i), \quad 0\leq k,j\leq m.
\end{align*}
Moreover, let $\signal_i^{(k)}:=\signal_i^{(k,k)}$, for all $0\leq k\leq m$.

    First we replace the randomized signal recommendation with its certainty-equivalent signal $\signal_i^{(k)}$ for each type $\hat{\type}^{(k)}_i$ while maintaining the same allocation rule.
We next show that when type $\hat{\type}^{(k-1)}_i$ misreports as the adjacent type $\hat{\type}^{(k)}_i$, his expected utility is weakly lower after the replacement than the one before the replacement.

Type $\hat{\type}^{(k-1)}_i$'s utility for misreporting as type $\hat{\type}^{(k)}_i$ in the direct mechanism is at least\footnote{Notice that the agent can potentially adopt double deviation strategies in the direct mechanism based on the signal realizations, i.e., after misreporting, the agent can walk away based on the signal realization.} 
\begin{align}\label{lower rent}
\util_i\rbr{ \hat{\type}^{(k-1)}_i;G_i^{(k)}} 
&=\util_i\rbr{\hat{\type}^{(k-1)}_i;H^{(k)}_i, \signal_i^{(k,k-1)}}
\geq \util_i\rbr{\hat{\type}^{(k-1)}_i;H^{(k)}_i, \signal_i^{(k)}},
\end{align}
where the equality follows from the definition of the certainty equivalent signal,
and the inequality holds since (1) $\signal_i^{(k,j)}$ is increasing in $j$ by \cref{asp:monotone_certainty_equivalence}; and (2) the utility function is weakly decreasing in the signal.
 By incentive compatibility in the original mechanism, we have
\begin{align}\label{IC in original mech}
\util_i\rbr{ \hat{\type}^{(k-1)}_i} :=\util_i\rbr{ \hat{\type}^{(k-1)}_i;G_i^{(k-1)}} 
\geq\util_i\rbr{\hat{\type}^{(k-1)}_i;G_i^{(k)}}.
\end{align}
Combining \eqref{lower rent} and \eqref{IC in original mech}, we conclude that
\begin{align}\label{local IC in certainty equivalence}
\util_i\rbr{\hat{\type}^{(k-1)}_i}
\geq \util_i\rbr{\hat{\type}^{(k-1)}_i;H^{(k)}_i, \signal_i^{(k)}}.
\end{align}
That is, type $\hat{\type}^{(k-1)}_i$ has no incentive to deviate to the higher adjacent type after we replace the signal recommendation with the certainty equivalent signal of the corresponding stochastic recommendation in the original mechanism. 
Because each agent's utility admits an expected utility representation, we have $\util_i(\type;G_{\allocspace},\signal):= \alloc \cdot\expostutil_i\rbr{1,\signal,\type}+ (1- \alloc )\cdot\expostutil_i\rbr{0,\signal,\type}=\util_i(\type;\alloc,\signal)$, where $\alloc$ is the mean of $G_{\allocspace}$.
Hence \eqref{local IC in certainty equivalence} implies \eqref{eq:incentive}.
\end{proof}

\subsection{Continuous Types}\label{subapx:continuous}

\providecommand{\mesh}{\operatorname{mesh}}
In this section, we prove \cref{thm:monotone optimal general} with a continuous type space. 
We can utilize our constructive proof in the finite type space to construct a superior candidate non-coordination mechanism in the continuous type space. The key challenge is to construct the signal recommendation function, since the inductive argument no longer works in the continuous type space. To circumvent this, our idea is to use step functions to approximate the signal recommendation function (\textbf{Steps 1 and 2}). When we discretize the type space, we can construct the signal recommendation function on the discretized grids as in the proof of Theorem 1 (\textbf{Step 1}). For every type that is not on the discretized grids, we can assign the signal recommendation of the rightmost grid point below this type as its signal recommendation. 
This way, we can construct a sequence of increasing step functions as an approximation to the signal recommendation function in the continuous type space, which admits a converging subsequence (\textbf{Step 2}). Correspondingly, we can also obtain the limit of the utility functions in this construction. This limit gives us a candidate mechanism. It remains to check whether incentive compatibility holds in the limit. Given the smoothness in the utility functions, we can ensure that IC holds in the limit (\textbf{Step 3}). This completes our construction of a candidate mechanism, the limit of the discretized approximation, which is an incentive-compatible non-coordination mechanism that outperforms any given non-coordination one.

Now, let us begin our formal proof.
Since both $\alloc_i(\cdot)$ and $\util_i(\cdot)$ are weakly increasing, each has at most countably many discontinuity points. We first define the discretizations based on the discontinuity points.

For each $m\ge 1$, choose a finite set
$\Theta_i^{(m)}=\{\hat{\type}^{(m,0)}_i<\hat{\type}^{(m,1)}_i<\cdots<\hat{\type}^{(m,K_m)}_i\}\subseteq\Theta_i$
such that $\Theta_i^{(m)}\subseteq \Theta_i^{(m+1)}$, and for any $m$,
\begin{enumerate}
\item $\max_{k=1,\dots,K_m}\big(\hat{\type}^{(m,k)}_i-\hat{\type}^{(m,k-1)}_i\big)\le \frac{1}{m}$;
\item every discontinuity point of $\alloc_i(\cdot)$ and $\util_i(\cdot)$ belongs to $\Theta_i^{(m)}$ for all sufficiently large $m$
(e.g.\ enumerate discontinuities and ensure the first $m$ are in $\Theta_i^{(m)}$).
\end{enumerate}
Define the left-rounding map $r_i^{(m)}:\Theta_i\to\Theta_i^{(m)}$ by
\[
r_i^{(m)}(\type_i):=\max\{\hat{\type}_i\in\Theta_i^{(m)}:\ \hat{\type}_i\le \type_i\}.
\]
Then $r_i^{(m)}(\type_i)\rightarrow \type_i$ as $m\rightarrow\infty$ and, by construction of $\Theta_i^{(m)}$, as $m\rightarrow\infty$, we have for every $\type_i\in\Theta_i$:
\begin{equation}\label{eq:rounding_convergence_A_i}
\alloc_i(r_i^{(m)}(\type_i))\to \alloc_i(\type_i),
\qquad
\util_i(r_i^{(m)}(\type_i))\to \util_i(\type_i).
\end{equation}

\paragraph{Step 1: Constructing a deterministic mechanism on the discretized grids.}
Restrict $(\alloc_i,\util_i)$ to the finite set $\Theta_i^{(m)}$.
Apply the finite-type version of \Cref{thm:monotone optimal general} to construct an alternative mechanism $(\alloc_i,\util_i^\dagger)$. Hence, we obtain a deterministic signal recommendation function
\[
\signal_i^{(m)}:\Theta_i^{(m)}\to \mathbb R_+ \quad \text{(weakly increasing in $\hat{\type}_i$)}
\]
such that for all $\hat{\type}_i,\hat{\type}'_i\in\Theta_i^{(m)}$,
\begin{align}
\util_i^\dagger\big(\hat{\type}_i;\alloc_i(\hat{\type}_i),\signal_i^{(m)}(\hat{\type}_i)\big)\ &\ge\ \util_i(\hat{\type}_i), \label{eq:grid_dom_A_i}\\
\util_i^\dagger\big(\hat{\type}_i;\alloc_i(\hat{\type}_i),\signal_i^{(m)}(\hat{\type}_i)\big)\ &\ge\
\util_i^\dagger\big(\hat{\type}_i;\alloc_i(\hat{\type}'_i),\signal_i^{(m)}(\hat{\type}'_i)\big). \label{eq:grid_IC_A_i}
\end{align}

\paragraph{Step 2: Extending the construction to the whole type space.}
Extend $\signal_i^{(m)}$ to all $\Theta_i$ by
\[
\signal_i^{(m)}(\type_i):=\signal_i^{(m)}(r_i^{(m)}(\type_i)),
\]
which is weakly increasing in $\type_i$. 
To extract a pointwise convergent subsequence in $\{\signal_i^{(m)}(\type_i)\}_m$, consider a bounded monotone transform
$\phi(\signal):=\frac{\signal}{1+\signal}\in[0,1]$.
Then $\phi\circ \signal_i^{(m)}$ is weakly increasing and bounded in $[0,1]$.
By Helly's selection theorem, there exists a subsequence (still indexed by $m$) and a weakly increasing function $\sigma_i:\Theta_i\to[0,1]$ such that as $m\rightarrow\infty$,
\[
\phi(\signal_i^{(m)}(\type_i))\to \sigma_i(\type_i)\qquad \forall \type_i\in\Theta_i.
\]
Define the pointwise limit of the signal recommendation as
\[
\signal_i(\type_i):=\phi^{-1}(\sigma_i(\type_i))=\frac{\sigma_i(\type_i)}{1-\sigma_i(\type_i)}.
\]
Under our maintained model assumption that certainty equivalents exist as finite signals, we may take the subsequence so that $\sigma_i(\type_i)<1$ for all $\type_i$; hence $\signal_i(\type_i)\in\mathbb R_+$.

\paragraph{Step 3: Verifying IC and utility dominance in the limit.}
Fix the converging subsequence identified in Step 2 and relabel it as $m$. 
Fix $\type_i,\type_i'\in\Theta_i$ and set $\hat{\type}_{i,m}:=r_i^{(m)}(\type_i)$ and $\hat{\type}'_{i,m}:=r_i^{(m)}(\type_i')$.
Apply \eqref{eq:grid_IC_A_i} with $(\hat{\type}_i,\hat{\type}'_i)=(\hat{\type}_{i,m},\hat{\type}'_{i,m})$:
\[
\util_i^\dagger\big(\hat{\type}_{i,m};\alloc_i(\hat{\type}_{i,m}),\signal_i^{(m)}(\hat{\type}_{i,m})\big)
\ \ge\
\util_i^\dagger\big(\hat{\type}_{i,m};\alloc_i(\hat{\type}'_{i,m}),\signal_i^{(m)}(\hat{\type}'_{i,m})\big).
\]
Take $m\to\infty$.
Using $\hat{\type}_{i,m}\to\type_i$, $\hat{\type}'_{i,m}\to\type_i'$, the convergence in \eqref{eq:rounding_convergence_A_i},
the observation that 
$\signal_i^{(m)}(\hat{\type}_{i,m})\to \signal_i(\type_i)$,
$\signal_i^{(m)}(\hat{\type}'_{i,m})\to \signal_i(\type_i')$,
and the continuity of the utility function, we obtain
\begin{equation}\label{eq:limit_IC_A_i}
\util_i^\dagger\big(\type_i;\alloc_i(\type_i),\signal_i(\type_i)\big)\ \ge\ \util_i^\dagger \big(\type_i;\alloc_i(\type_i'),\signal_i(\type_i')\big)
\qquad \forall \type_i,\type_i'\in\Theta_i.
\end{equation}
Thus, the mechanism $\type_i\mapsto(\alloc_i(\type_i),\signal_i(\type_i))$ is incentive compatible.

Similarly, apply \eqref{eq:grid_dom_A_i} at $\hat{\type}_i=\hat{\type}_{i,m}$ and take $m\to\infty$ to get
\begin{equation}\label{eq:limit_dom_A_i}
\util_i^\dagger\big(\type_i;\alloc_i(\type_i),\signal_i(\type_i)\big)\ \ge\ \util_i(\type_i)\qquad \forall \type_i\in\Theta_i.
\end{equation}
Equations \eqref{eq:limit_IC_A_i}--\eqref{eq:limit_dom_A_i} complete the continuous-type extension.

\section{Proof of \cref{prop:implement_via_majorization}}
\label{apx:discussion}

We first recap the definition and the characterizations of feasible interim allocation rules from \citet{border1991implementation} and \citet{che2013generalized}.  
\begin{definition}[feasibility]
Given any interim allocation rule $\allocs=(\alloc_i)_{i=1}^n$ where $\alloc_i:\Theta_i\rightarrow [0,1]$, 
$\allocs$ is \emph{feasible} if there exists an ex-post allocation rule $\contestallocs$
such that 
$\alloc_i(\type_i) = \expect[\type_{-i}]{\contestalloc_i(\type_i, \type_{-i})}$ for any agent~$i$ and type $\type_i$.
We say an interim allocation--utility pair $(\allocs,\utils)$ is \emph{feasible} if $\allocs$ is feasible. 
\end{definition}

\begin{lemma}[\citealp{che2013generalized}]
\label{lem:border}
Given a set $\boldsymbol{A}=\prod_{i=1}^nA_i\subset \typespace$, let $w(\types,\boldsymbol{A}) = |\lbr{i : \type_i\in A_i}|$ be the number of agents whose type $\type_i$ is in $A_i$.\footnote{Here, $|\cdot|$ denotes the cardinality of a set.}
The interim allocation rule $\allocs$ is interim feasible if and only if 
\vspace{-8pt}
\begin{align*}
\tag{IF}\label{eq:feasibility}
\sum_i \int_{A_i}  \alloc_i(\type_i) \dd \dist_i(\type_i)
\leq \int_{A} \min\lbr{k,w(\types,\boldsymbol{A})} \dd \dists(\type)
\qquad \forall \boldsymbol{A}=\prod_{i=1}^nA_i\subset \typespace.
\end{align*}
Moreover, for monotone allocations in symmetric environments, 
\eqref{eq:feasibility} is equivalent to the following:\footnote{In a symmetric environment, by a slight abuse of notation, we use $\dist=\dist_i$ for all $i$ to denote each agent's type distribution.}
\begin{align}\label{eq:feasibility symmetric}
\tag{$\widehat{\mathrm{IF}}$}
\int_{\type}^{\bar{\type}} \alloc(z)\dd \dist(z) \leq \int_{\type}^{\bar{\type}} \efficientalloc(z) \dd\dist(z), \quad\forall \type\in[\underline{\type},\bar{\type}],
\end{align}
where $\efficientalloc(\type)=\sum_{j=0}^{k-1}\binom{n-1}{j} 
\cdot (1-\dist(\type))^j\cdot \dist^{n-1-j}(\type)$ is the interim allocation rule for allocating $k$ items efficiently. 
\end{lemma}

\begin{proof}[Proof of \cref{prop:implement_via_majorization}]
    For any symmetric interim allocation--utility pair $(\alloc,\util)$ 
that is implementable by a non-coordination mechanism, 
by definition, there exists a mapping from the signal space to the allocation space $\signalalloc(\reportsignal)$ specifying the allocation for each agent given the generated signal $\reportsignal$.\footnote{Notice that such a mapping might not exist if $(\alloc,\util)$ is implementable by a direct mechanism that is not a non-coordination mechanism.}
The distribution over types $\dist(\type)$ induces a distribution over signals; call it $\hat\dist(\signal)$.
Similarly, a feasibility constraint on the allocation rule defined in the signal space may be induced by \eqref{eq:feasibility symmetric} based on $\hat{\dist}$.
Such operations are valid since the recommended signal is a non-decreasing function of the type. 
By Theorems 1 and 2 in \citet{KMS2021extreme}, 
any monotone feasible allocation rule
$\signalalloc$ can be written as a convex combination of the extreme points. 
Using the construction of the extreme points in Theorem 3 of \citet{KMS2021extreme}, one can easily verify that the extreme points are the coarse ranking contest rules defined in \cref{def: coarse ranking contest rule}. 
Hence $\signalalloc$ can be expressed as a randomization over the coarse ranking contest rules.
Notice that the above operations do not affect the agents' incentives by preserving the distribution over outcomes; hence the expected utility of each agent in the coarse ranking contest is still~$\util$. 
\end{proof}

\section{Proof of \cref{prop:monotone optimal general}}
\label{apx:suboptimality}
Given any utility function of the agents, consider a type distribution where $H_i$ has binary support for every agent $i$, denoted by $\{\underline{\type}_i,\bar{\type}_i\}$ where $\underline{\type}_i < \bar{\type}_i$. 
Given any strictly increasing interim allocation $\alloc_i$, let $\underline{\signal}_i $ ($ \bar{\signal}_i$) be the signal recommended to type $\underline{\type}_i$ (type $\bar{\type}_i$) in the non-coordination mechanism. 
It is easy to verify that in the optimal non-coordination mechanism, $\underline{\signal}_i = 0$ and $\underline{\signal}_i< \bar{\signal}_i$.

Since the utility function is continuous and strictly decreasing in $s$ and $\bar{\signal}_i$ is strictly positive, 
there exists a non-degenerate distribution $G_i$ over signal recommendations such that an agent with type $\bar{\type}_i$ is indifferent between signal recommendation $G_i$ and deterministic signal recommendation $\bar{\signal}_i$.\footnote{Such a non-degenerate distribution $G_i$ exists since $\bar{\signal}_i > 0$. We can construct it via a distribution with binary support $\bar{\signal}_i-\epsilon$ and $\bar{\signal}_i+\epsilon$, where $\epsilon\in(0,\bar{\signal}_i)$ and the probability of $\bar{\signal}_i+\epsilon$ is $q_{\epsilon}$. 
The strict monotonicity of utility in costly signals implies that $q_{\epsilon}$ must be an interior solution, i.e., $q_{\epsilon}\in (0,1)$. 
}
Moreover, by \cref{asp:decreasing_certainty_equivalence}, the utility of type $\underline{\type}_i$ for misreporting as type $\bar{\type}_i$ is strictly lower given $G_i$ compared to deterministic signal recommendation $\bar{\signal}_i$. 
Therefore, there exists $G'_i$ that first-order stochastically dominates $G_i$ such that type $\underline{\type}_i$ obtains the same utility from signal recommendation $G'_i$ and deterministic signal recommendation $\bar{\signal}_i$. 
By offering signal recommendation $\underline{\signal}_i$ to type $\underline{\type}_i$ with interim allocation $\alloc_i(\underline{\type}_i)$, 
and offering randomized signal recommendation $G'_i$ to type $\bar{\type}_i$ with interim allocation $\alloc_i(\bar{\type}_i)$ in the general mechanism, 
the utility of all types weakly improves, and the utility of type $\bar{\type}_i$ strictly improves. 

\section{Proof of \cref{thm:convex_monotone}}\label{appx B: general allocations}

To simplify the exposition, we assume that the  type space of agent $i$ is discrete and finite; that is, $\Theta_i=\{\hat{\type}^{(0)}_i,\dots , \hat{\type}^{(m)}_i\}$, where $\hat{\type}^{(0)}_i < \dots < \hat{\type}^{(m)}_i$.

Let $\allocs\primed$ be a monotone rearrangement of $\allocs$ that is feasible.
We will construct a $\utils\primed$ such that $(\allocs\primed,\utils\primed)$ is implementable by a non-coordination mechanism 
and show that $\util_i\primed(\hat{\type}^{(k)}_i)\geq \util_i(\hat{\type}^{(k)}_i)$ for all $i$ and all $k\in \{0,\dots,m\}$. 

When $k=0$, let
$\hat\signal_i^{(0)} := 0$.
Let $\util\primed_i(\hat{\type}^{(0)}_i)= \util_i(\hat{\type}^{(0)}_i,\alloc\primed_i(\hat{\type}^{(0)}_i), \hat\signal_i^{(0)})$.
Since $\allocs\primed$ is a monotone rearrangement of $\allocs$, we have $\alloc\primed_i(\hat{\type}^{(0)}_i) \geq \min_{\hat{\type}_i\in\Theta_i}\alloc_i(\hat{\type}_i)$.
Therefore, 
\begin{align*}
\util\primed_i(\hat{\type}^{(0)}_i)\geq \min_{\hat{\type}_i\in\Theta_i}\util_i(\hat{\type}_i,\alloc\primed_i(\hat{\type}^{(0)}_i), 0)
\geq \min_{\hat{\type}_i\in\Theta_i}\util_i(\hat{\type}_i)
\geq \util_i(\hat{\type}^{(0)}_i)
\end{align*}
where the last inequality holds since the interim utility must be monotone increasing in the type of the agent and $\hat{\type}^{(0)}_i$ is the lowest type.

Now consider any $k\geq 1$. Suppose $\alloc\primed_i(\hat{\type}^{(k)}_i) \geq \alloc_i(\hat{\type}^{(k)}_i)$.  We can apply the same argument as in \cref{thm:monotone optimal general} to construct $\util\primed_i(\hat{\type}^{(k)}_i)$.

Suppose alternatively $\alloc\primed_i(\hat{\type}^{(k)}_i) < \alloc_i(\hat{\type}^{(k)}_i)$.
Let $\bar\signal^{(k)}_i$ be the signal such that type $\hat{\type}^{(k-1)}_i$ is indifferent between truthfully reporting and misreporting as type $\hat{\type}^{(k)}_i$:
\begin{align*}
    \util\primed_i(\hat{\type}^{(k-1)}_i) = \util_i(\hat{\type}^{(k-1)}_i;\alloc\primed_i(\hat{\type}^{(k)}_i),\bar\signal^{(k)}_i).
\end{align*}
Let $\bar\signal^{(k)}_i$  be the signal recommended to type $\hat{\type}^{(k)}_i$. 
Hence $\util\primed_i(\hat{\type}^{(k)}_i) = \util_i(\hat{\type}^{(k)}_i;\alloc\primed_i(\hat{\type}^{(k)}_i),\bar\signal^{(k)}_i).$

Following similar arguments as in \cref{thm:monotone optimal general}, we can show that $(\allocs\primed,\utils\primed)$ can be implemented by a direct non-coordination mechanism.
It remains to show that $\util\primed_i(\hat{\type}^{(k)}_i)
\geq  \util_i(\hat{\type}^{(k)}_i)$.

Since $\allocs\primed$ is a monotone rearrangement of $\allocs$, for any $i$ and any type $\hat{\type}^{(k)}_i$, there must exist $k' > k$ such that  $\alloc_i(\hat{\type}^{(k')}_i)\leq \alloc\primed_i(\hat{\type}^{(k)}_i) $.
Let  $\tilde{\signal}_i$ be the signal such that type $\hat{\type}^{(k-1)}_i$ is indifferent between truthfully reporting in the direct non-coordination mechanism and consuming $(\alloc_i(\hat{\type}^{(k')}_i),\tilde{\signal}_i)$:
\begin{align}\label{eq: def of tilde s}
    \util\primed_i(\hat{\type}^{(k-1)}_i) = \util_i(\hat{\type}^{(k-1)}_i;\alloc_i(\hat{\type}^{(k')}_i),\tilde{\signal}_i).
\end{align}
Since type $\hat{\type}^{(k-1)}_i$ is indifferent between these two options, $\alloc\primed_i(\hat{\type}^{(k)}_i) \geq \alloc_i(\hat{\type}^{(k')}_i)$ implies that
$\bar\signal^{(k)}_i \geq \tilde{\signal}_i$.
Moreover,
by \cref{asp:weak_single_crossing} (wSCP), we have that for type $\hat{\type}^{(k)}_i$,
\begin{equation}\label{eq:appendix B last >}
    \util_i(\hat{\type}^{(k)}_i;\alloc\primed_i(\hat{\type}^{(k)}_i),\bar\signal^{(k)}_i)
\geq \util_i(\hat{\type}^{(k)}_i;\alloc_i(\hat{\type}^{(k')}_i),\tilde{\signal}_i).
\end{equation}

Let $\dista^{(k)}_i$ be the distribution over signal recommendations to type $\hat{\type}^{(k)}_i$ under the original mechanism that gives agent $i$ interim utility $\util_i$. 
We first establish the following two inequalities that describe the preference of type $\hat{\type}^{(k')}_i$ and type $\hat{\type}^{(k-1)}_i$.

\begin{claim}\label{claim:k'}
    Type $\hat{\type}^{(k')}_i$ prefers receiving the deterministic recommendation $(\alloc_i(\hat{\type}^{(k')}_i),\tilde{\signal}_i)$ to  $(\alloc_i(\hat{\type}^{(k)}_i),\dista^{(k)}_i)$ in the original direct mechanism:
$$\util_i(\hat{\type}^{(k')}_i;\alloc_i(\hat{\type}^{(k')}_i),\tilde{\signal}_i)
\geq \util_i(\hat{\type}^{(k')}_i;\alloc_i(\hat{\type}^{(k)}_i),\dista^{(k)}_i).$$
\end{claim}

\begin{claim}\label{claim:k-1}
    Type $\hat{\type}^{(k-1)}_i$ prefers receiving the deterministic recommendation $(\alloc_i(\hat{\type}^{(k')}_i),\tilde{\signal}_i)$ to  $(\alloc_i(\hat{\type}^{(k)}_i),\dista^{(k)}_i)$ in the original direct mechanism:
$$\util_i(\hat{\type}^{(k-1)}_i;\alloc_i(\hat{\type}^{(k')}_i),\tilde{\signal}_i)
\geq \util_i(\hat{\type}^{(k-1)}_i;\alloc_i(\hat{\type}^{(k)}_i),\dista^{(k)}_i).$$
\end{claim}

Since $\hat{\type}^{(k-1)}_i< \hat{\type}^{(k)}_i <\hat{\type}^{(k')}_i$, 
by \cref{asp:no_concave_crossing}, we know that type $\hat{\type}^{(k)}_i$ also prefers receiving the deterministic recommendation $(\alloc_i(\hat{\type}^{(k')}_i),\tilde{\signal}_i)$ to $(\alloc_i(\hat{\type}^{(k)}_i),\dista^{(k)}_i)$ in the original direct mechanism:
\begin{equation}\label{eq: k preference}
    \util_i(\hat{\type}^{(k)}_i;\alloc_i(\hat{\type}^{(k')}_i),\tilde{\signal}_i)
\geq \util_i(\hat{\type}^{(k)}_i;\alloc_i(\hat{\type}^{(k)}_i),\dista^{(k)}_i).
\end{equation}

\eqref{eq:appendix B last >} and \eqref{eq: k preference}  imply that $\util_i(\hat{\type}^{(k)}_i;\alloc\primed_i(\hat{\type}^{(k)}_i),\bar\signal^{(k)}_i)\geq \util_i(\hat{\type}^{(k)}_i;\alloc_i(\hat{\type}^{(k)}_i),\dista^{(k)}_i)$.
Notice that by construction, the left-hand side is $\util\primed_i(\hat{\type}^{(k)}_i)$ and by definition, the right-hand side is $\util_i(\hat{\type}^{(k)}_i)$.
Hence, we can conclude that 
$\util\primed_i(\hat{\type}^{(k)}_i)
\geq  \util_i(\hat{\type}^{(k)}_i)$.

\begin{proof}[Proof of \cref{claim:k'}]
    By the inductive argument $\util\primed_i(\hat{\type}^{(k-1)}_i)
\geq \util_i(\hat{\type}^{(k-1)}_i)$.
By \eqref{eq: def of tilde s}, the definition of $\tilde{\signal}_i$, we have $\tilde{\signal}_i\leq \signal_i^{(k',k)}$.
Moreover, we already know that by \cref{asp:monotone_certainty_equivalence}, $\signal_i^{(k',k)}\leq\signal_i^{(k')}$.
These together imply that
\begin{align}\label{eq: to be checked}
\util_i(\hat{\type}^{(k')}_i;\alloc_i(\hat{\type}^{(k')}_i),\tilde{\signal}_i)
\geq 
\util_i(\hat{\type}^{(k')}_i).
\end{align}
By incentive compatibility in the original direct mechanism, we have
\begin{equation}
    \util_i(\hat{\type}^{(k')}_i) \geq \util_i(\hat{\type}^{(k')}_i;\alloc_i(\hat{\type}^{(k)}_i),\dista^{(k)}_i).
\end{equation}
Combining the two, we have the desired inequality.
\end{proof}

\begin{proof}[Proof of \cref{claim:k-1}]

This is because 
\begin{align*}
\util_i(\hat{\type}^{(k-1)}_i;\alloc_i(\hat{\type}^{(k')}_i),\tilde{\signal}_i)
= \util\primed_i(\hat{\type}^{(k-1)}_i)
\geq \util_i(\hat{\type}^{(k-1)}_i)
\geq \util_i(\hat{\type}^{(k-1)}_i;\alloc_i(\hat{\type}^{(k)}_i),\dista^{(k)}_i).
\end{align*}
The first inequality holds by the induction assumption. 
The second inequality holds because $\util_i(\hat{\type}^{(k-1)}_i;\alloc_i(\hat{\type}^{(k)}_i),\dista^{(k)}_i)$
is the utility that type $\hat{\type}^{(k-1)}_i$ obtains by deviating to report type $\hat{\type}^{(k)}_i$ and always following the signal recommendation. 
\end{proof}

\section{Non-separable Utilities}
\label{apx:non_separable}

In this section, we clarify why \cref{asp:monotone_certainty_equivalence} may fail under most non-separable preferences.
Then we propose a weaker
certainty-equivalence monotonicity requirement. We show that under this weaker requirement, we can extend the comparison between any non-coordination implementation and a simple form of coordination implementations. 
A similar analysis of ranking the performance of simple mechanisms has also been investigated in the auction context under different objectives and utility models \citep[e.g.,][]{milgrom1982theory,matthews1987comparing}.

Recall that \cref{asp:monotone_certainty_equivalence} requires that for \emph{every} distribution $G$ over allocations and recommended
signals, the certainty-equivalent signal $\Sigma(G,\theta)$ is increasing in $\theta$.
This can fail even for commonly used non-separable expected-utility
representations. 
For example, under the preference that exhibits constant absolute risk aversion: $\expostutil_i(\signalalloc_i,\signal_i,\type_i) = 1-\exp(-a (\type_i\cdot \signalalloc_i - \signal_i))$ for some constant $a > 0$, \cref{asp:monotone_certainty_equivalence} is violated. 

However, we highlight that for some design questions, one does not need ICE for \emph{all} lotteries. Below, we propose a weaker monotonicity property that is sufficient to compare
\emph{non-coordination} (all-pay type formats) and \emph{ex post certification} (winner-only formats).

\subsection{Ex Post Certification Mechanisms}\label{app:nonsep:certification}

We formalize a subclass of direct mechanisms where \emph{only tentative winners} are ever
asked to produce a positive signal.

\begin{definition}[Ex post certification mechanism]\label{def:cert}
A direct mechanism $(\signalrecommends,\mechallocs)$ is an \emph{ex post certification mechanism} if there exist
a \emph{tentative winner selection rule} $\hat y:\Theta\to\allocspace$ and a
\emph{certification schedule} $\hat s:\Theta_i\to\mathbb R_+$ such that the following holds.

\begin{enumerate}
\item For every report profile $\hat\theta\in\Theta$, the mechanism draws a tentative winner
set $W$ according to $\hat y(\hat\theta)$.

\item Given $(\hat\theta,W)$, the recommendation to agent $i$ is
\[
\tilde s_i(\hat\theta,W)=
\begin{cases}
\hat s(\hat\theta_i), &\text{if } i\in W,\\
0, &\text{if } i\notin W.
\end{cases}
\]
Thus, each agent is asked to certify only when he is selected as a tentative winner.

\item The allocation rule awards the item to $i$ if and only if $i$ is a tentative winner and he
follows the recommendation, i.e., 
$y_i(\hat\theta,W,s') =
\mathbf 1\{i\in W\}\cdot \mathbf 1\{s'_i=\hat s(\hat\theta_i)\}$.
\end{enumerate}
\end{definition}

The above representation captures the idea that the designer first selects tentative
winners (potentially using cross-agent information), and then requires only tentative winners
to produce costly certificates. This is the natural analog of ``winner-pays'' formats in
standard auction theory, whereas non-coordination mechanisms correspond to ``all-pay'' formats.

\subsection{A Weaker Monotonicity Condition}\label{app:nonsep:WICE}

Note that in any ex post certification mechanism, each agent faces a lottery over signals. This lottery has binary support, and one potential signal is $0$. Therefore, to show the superiority of non-coordination mechanisms over ex post certification mechanisms, it suffices to consider binary lotteries before converting them to certainty-equivalent signals. 
This motivates our following definition. 

\begin{assumption}[Winner-only Increasing Certainty Equivalence (WICE)]\label{ass:WICE}
A utility function $\expostutil_i$ is said to have (strict) winner-only increasing certainty equivalence if for any distribution $G$ over allocations and signals with binary support on $(0,0)$ and $(1,s)$ for some $s\in\reals_+$, his certainty equivalent signal $\ce_i(G,\type)$ is (strictly) increasing in $\type$. 
\end{assumption}

Assumption~\ref{ass:WICE} is strictly weaker than ICE (\cref{asp:monotone_certainty_equivalence}) because it only considers \emph{binary} lotteries rather than any arbitrary distribution over allocations and signals.
To better understand \cref{ass:WICE}, we provide a simple sufficient condition on the primitives such that \cref{ass:WICE} holds. 
We emphasize that this condition is \emph{not} meant to be necessary for all non-separable
utility functions; rather, it is a transparent sufficient condition tailored to the ex post
certification environment to facilitate our understanding.

\begin{lemma}[Sufficient Condition for WICE]\label{lem:WICE-sufficient-simple}
Suppose \cref{asp:monotonicity} holds. In addition, assume:
\begin{enumerate}
\item there exists a weakly decreasing function $u^0(s)$ such that $u(0,s,\type)=u^0(s)$ for all $(s,\type)$;
\item the utility function satisfies decreasing differences in $(\type,s)$ when $x=1$, i.e., 
for any $\type'>\type$ and any $s'>s$,
\begin{equation}\label{eq:DD}
u(1,s',\type')-u(1,s,\type') \;\le\; u(1,s',\type)-u(1,s,\type).
\end{equation}
\end{enumerate}
Then the utility function satisfies \cref{ass:WICE} (WICE).
\end{lemma}
The first condition in \cref{lem:WICE-sufficient-simple} states that when an agent does \emph{not} receive the
resource ($x=0$), his utility from producing a costly certificate $s$ is purely a type-independent
``wasted effort'' term. 
This is realistic in applications where the type $\type$ captures
the valuation of the resource, because the type is only relevant for utility if the agent receives the resource.
The second condition in \cref{lem:WICE-sufficient-simple} 
is a standard \emph{submodularity} (or decreasing-differences)
condition. Economically, it states that a higher certificate requirement compresses payoff
differences across types in the winning state. \footnote{If $u$ is $C^2$ in $(s,\type)$, then \eqref{eq:DD} is equivalent to
$u_{\type s}(1,s,\type)\le 0$ for all $(s,\type)$.}
Put differently, a higher certificate reduces the incremental utility
of having a higher type.

\begin{proof}[Proof of \cref{lem:WICE-sufficient-simple}]
Fix any binary distribution $G$ with support $\{(0,0),(1,\bar s)\}$ for some $\bar s\ge 0$,
and let $q:=\Pr_G(x=1)\in[0,1]$. If $q\in\{0,1\}$ or $\bar s=0$, then $\ce(G,\type)$ is
constant in $\type$ and the claim is immediate. Hence assume $q\in(0,1)$ and $\bar s>0$.

Let $c(\type):=\ce(G,\type)$. By definition, $c(\type)$ satisfies
\begin{equation}\label{eq:CE-binary}
(1-q)\,u^0\big(c(\type)\big)+q\,u\big(1,c(\type),\type\big)
=(1-q)\,u^0(0)+q\,u(1,\bar s,\type).
\end{equation}
Rearranging \eqref{eq:CE-binary} gives
\begin{equation}\label{eq:c-leq-sbar}
q\Big(u(1,c(\type),\type)-u(1,\bar s,\type)\Big)
=(1-q)\Big(u^0(0)-u^0(c(\type))\Big)\ \ge\ 0,
\end{equation}
where the inequality uses that $u^0(\cdot)$ is (weakly) decreasing. Since $u(1,\cdot,\type)$
is (weakly) decreasing, \eqref{eq:c-leq-sbar} implies $c(\type)\le \bar s$.

Now fix $\type'>\type$ and set $c:=c(\type)$. Consider the
certainty-equivalence equation for type $\type'$ evaluated at $c$:
\begin{align*}
&(1-q)\,u^0(c)+q\,u(1,c,\type')-(1-q)\,u^0(0)-q\,u(1,\bar s,\type') \\
&\qquad=
\underbrace{\Big((1-q)\,u^0(c)+q\,u(1,c,\type)-(1-q)\,u^0(0)-q\,u(1,\bar s,\type)\Big)}_{=\,0 \text{ by \eqref{eq:CE-binary}}}\\
&\qquad\quad
+\,q\Big(\big[u(1,c,\type')-u(1,c,\type)\big]-\big[u(1,\bar s,\type')-u(1,\bar s,\type)\big]\Big).
\end{align*}
By the decreasing-differences condition in the winner state, the type
premium $u(1,s,\type')-u(1,s,\type)$ is weakly decreasing in $s$. Since $c\le \bar s$, we have
\[
u(1,c,\type')-u(1,c,\type)\ \ge\ u(1,\bar s,\type')-u(1,\bar s,\type),
\]
so the last display is weakly nonnegative. Hence
\[
(1-q)\,u^0(c)+q\,u(1,c,\type') \ \ge\ (1-q)\,u^0(0)+q\,u(1,\bar s,\type').
\]
Because $u^0(\cdot)$ and $u(1,\cdot,\type')$ are decreasing, the map
$s\mapsto (1-q)u^0(s)+q u(1,s,\type')$ is decreasing, and the solution
$c(\type')$ to \eqref{eq:CE-binary} for type $\type'$ must satisfy $c(\type')\ge c(\type)$.
This proves WICE.
\end{proof}

\subsection{Comparison of Mechanisms}
Building on the definitions and techniques from this paper, we provide a transparent comparison between non-coordination mechanisms and ex post certification mechanisms under the weaker condition \cref{ass:WICE}, which is satisfied by a large class of non-separable utilities (\cref{lem:WICE-sufficient-simple}).

\begin{proposition}[Non-coordination vs Certification]
\label{prop:vs_certification}
Suppose Assumptions~\ref{asp:monotonicity}, \ref{asp:weak_single_crossing}, and \ref{ass:WICE} hold.
Fix any interim allocation--utility pair $(\allocs,\utils)$ such that $(\allocs,\utils)$ can be implemented by an ex post certification mechanism and $\alloc_i$ is weakly increasing for all $i$. 
Then there exists an interim allocation--utility pair $(\allocs\primed,\utils\primed)$ that can be implemented by a direct non-coordination mechanism
such that $\allocs\primed=\allocs$, and 
$
\util\primed_i(\type_i) \geq \util_i(\type_i)$ for all $ i, \type_i$. 
\end{proposition}
The proof of \cref{prop:vs_certification} is almost identical to the proof of \cref{thm:monotone optimal general}, with the exception that, here, it suffices to consider distributions over signals and allocations that have binary support from each agent's perspective. 
The details are omitted here to avoid repetition.

\newpage
\bibliographystyle{apalike}
\bibliography{references}

@article{mcafee1992bidding,
  title={Bidding rings},
  author={McAfee, R Preston and McMillan, John},
  journal={American Economic Review},
  pages={579--599},
  year={1992},
  publisher={JSTOR}
}

@article{border1991implementation,
  title={Implementation of reduced form auctions: A geometric approach},
  author={Border, Kim C},
  journal={Econometrica},
  pages={1175--1187},
  year={1991},
  publisher={JSTOR}
}

@article{baye1996all,
  title={The all-pay auction with complete information},
  author={Baye, Michael R and Kovenock, Dan and De Vries, Casper G},
  journal={Economic Theory},
  volume={8},
  number={2},
  pages={291--305},
  year={1996},
  publisher={Springer}
}

@article{KMS2021extreme,
  title={Extreme points and majorization: Economic applications},
  author={Kleiner, Andreas and Moldovanu, Benny and Strack, Philipp},
  journal={Econometrica},
  volume={89},
  number={4},
  pages={1557--1593},
  year={2021},
  publisher={Wiley Online Library}
}

@article{FK2019muddled,
  title={Muddled information},
  author={Frankel, Alex and Kartik, Navin},
  journal={Journal of Political Economy},
  volume={127},
  number={4},
  pages={1739--1776},
  year={2019}
}

@article{che2013generalized,
  title={Generalized Reduced-Form Auctions: A Network-Flow Approach},
  author={Che, Yeon-Koo and Kim, Jinwoo and Mierendorff, Konrad},
  journal={Econometrica},
  volume={81},
  number={6},
  pages={2487--2520},
  year={2013},
  publisher={Wiley Online Library}
}

@article{Ball2019,
  title={Scoring strategic agents},
  author={Ball, Ian},
  journal={American Economic Journal: Microeconomics},
  year={2025},
  volume={17},
  number={1},
  pages={97--129},
  doi={10.1257/mic.20230275}
}

@inproceedings{hartline2008,
  title={Optimal mechanism design and money burning},
  author={Hartline, Jason D and Roughgarden, Tim},
  booktitle={Proceedings of the Fortieth Annual ACM Symposium on Theory of Computing},
  pages={75--84},
  year={2008}
}

@article{myerson1981,
  title={Optimal auction design},
  author={Myerson, Roger B},
  journal={Mathematics of Operations Research},
  volume={6},
  number={1},
  pages={58--73},
  year={1981},
  publisher={INFORMS}
}

@article{green1986partially,
  title={Partially verifiable information and mechanism design},
  author={Green, Jerry R and Laffont, Jean-Jacques},
  journal={The Review of Economic Studies},
  volume={53},
  number={3},
  pages={447--456},
  year={1986},
  publisher={Wiley-Blackwell}
}

@article{moldovanu2001optimal,
  title={The optimal allocation of prizes in contests},
  author={Moldovanu, Benny and Sela, Aner},
  journal={American Economic Review},
  volume={91},
  number={3},
  pages={542--558},
  year={2001}
}

@article{fang2020turning,
  title={Turning up the heat: The discouraging effect of competition in contests},
  author={Fang, Dawei and Noe, Thomas and Strack, Philipp},
  journal={Journal of Political Economy},
  volume={128},
  number={5},
  pages={1940--1975},
  year={2020},
  publisher={University of Chicago Press}
}

@article{siegel2009all,
  title={All-pay contests},
  author={Siegel, Ron},
  journal={Econometrica},
  volume={77},
  number={1},
  pages={71--92},
  year={2009},
  publisher={Wiley Online Library}
}

@article{lazearRosen1981rank,
  title={Rank-order tournaments as optimum labor contracts},
  author={Lazear, Edward P and Rosen, Sherwin},
  journal={Journal of Political Economy},
  volume={89},
  number={5},
  pages={841--864},
  year={1981},
  publisher={University of Chicago Press}
}

@article{skaperdas1996contest,
  title={Contest success functions},
  author={Skaperdas, Stergios},
  journal={Economic Theory},
  volume={7},
  number={2},
  pages={283--290},
  year={1996},
  publisher={Springer}
}

@article{perez2022test,
  title={Test design under falsification},
  author={Perez-Richet, Eduardo and Skreta, Vasiliki},
  journal={Econometrica},
  volume={90},
  number={3},
  pages={1109--1142},
  year={2022},
  publisher={Wiley Online Library}
}

@article{zhang2023optimal,
  title={Optimal contests with incomplete information and convex effort costs},
  author={Zhang, Mengxi},
  journal={Theoretical Economics},
  volume={19},
  number={1},
  pages={95--129},
  year={2024},
  publisher={Wiley Online Library}
}

@article{gershkov2022optimal,
  title={Optimal insurance: Dual utility, random losses, and adverse selection},
  author={Gershkov, Alex and Moldovanu, Benny and Strack, Philipp and Zhang, Mengxi},
  journal={American Economic Review},
  year={2023},
  volume={113},
  number={10},
  pages={2581--2614},
  doi={10.1257/aer.20221247}
}

@inproceedings{hardt2016strategic,
  title={Strategic classification},
  author={Hardt, Moritz and Megiddo, Nimrod and Papadimitriou, Christos and Wootters, Mary},
  booktitle={Proceedings of the 2016 ACM Conference on Innovations in Theoretical Computer Science},
  pages={111--122},
  year={2016}
}

@article{spence1973job,
  title={Job market signaling},
  author={Spence, Michael},
  journal={The Quarterly Journal of Economics},
  volume={87},
  number={3},
  pages={355--374},
  year={1973}
}

@article{li2026,
  title={Allocating Resources under Strategic Misrepresentation},
  author={Li, Yingkai and Qiu, Xiaoyun},
  journal={Working paper},
  year={2026}
}

@article{perez2023fraud,
  title={Fraud-proof non-market allocation mechanisms},
  author={Perez-Richet, Eduardo and Skreta, Vasiliki},
journal={Working paper},
  year={2023}
}

@article{perez2024score,
  title={Score-based mechanisms},
  author={Perez-Richet, Eduardo and Skreta, Vasiliki},
  journal={arXiv preprint arXiv:2403.08031},
  year={2024}
}

@article{condorelli2012money,
  title={What money can't buy: Efficient mechanism design with costly signals},
  author={Condorelli, Daniele},
  journal={Games and Economic Behavior},
  volume={75},
  number={2},
  pages={613--624},
  year={2012},
  publisher={Elsevier}
}

@article{chakravarty2013optimal,
  title={Optimal allocation without transfer payments},
  author={Chakravarty, Surajeet and Kaplan, Todd R},
  journal={Games and Economic Behavior},
  volume={77},
  number={1},
  pages={1--20},
  year={2013},
  publisher={Elsevier}
}

@article{maskin1984optimal,
  title={Optimal auctions with risk averse buyers},
  author={Maskin, Eric and Riley, John},
  journal={Econometrica},
  pages={1473--1518},
  year={1984},
  publisher={JSTOR}
}

@article{milgrom1994monotone,
  title={Monotone comparative statics},
  author={Milgrom, Paul and Shannon, Chris},
  journal={Econometrica},
  pages={157--180},
  year={1994},
  publisher={JSTOR}
}

@article{akbarpour2024redistributive,
  title={Redistributive allocation mechanisms},
  author={Akbarpour, Mohammad and Dworczak, Piotr and Kominers, Scott Duke},
  journal={Journal of Political Economy},
  volume={132},
  number={6},
  pages={1831--1875},
  year={2024},
  publisher={The University of Chicago Press Chicago, IL}
}

@article{pratt1976risk,
  title={Risk aversion in the small and in the large},
  author={Pratt, John W},
  journal={Econometrica},
  volume={32},
  number={1/2},
  pages={122--136},
  year={1964},
  publisher={JSTOR},
  doi={10.2307/1913738}
}

@article{myerson1982optimal,
  title={Optimal coordination mechanisms in generalized principal--agent problems},
  author={Myerson, Roger B},
  journal={Journal of Mathematical Economics},
  volume={10},
  number={1},
  pages={67--81},
  year={1982},
  publisher={Elsevier}
}

@article{yang2024comparison,
  title={Comparison of screening devices},
  author={Yang, Frank and Dworczak, Piotr and Akbarpour, Mohammad},
  journal={Journal of Political Economy, forthcoming},
  year={2026}
}

@article{finkelstein2019take,
  title={Take-up and targeting: Experimental evidence from {SNAP}},
  author={Finkelstein, Amy and Notowidigdo, Matthew J},
  journal={The Quarterly Journal of Economics},
  volume={134},
  number={3},
  pages={1505--1556},
  year={2019},
  publisher={Oxford University Press}
}

@article{quah2012aggregating,
  title={Aggregating the single crossing property},
  author={Quah, John K-H and Strulovici, Bruno},
  journal={Econometrica},
  volume={80},
  number={5},
  pages={2333--2348},
  year={2012},
  publisher={Wiley Online Library}
}

@article{kartik2024single,
  title={Single-crossing differences in convex environments},
  author={Kartik, Navin and Lee, SangMok and Rappoport, Daniel},
  journal={The Review of Economic Studies},
  volume={91},
  number={5},
  pages={2981--3012},
  year={2024},
  publisher={Oxford University Press UK}
}

@techreport{gao2005,
  author       = {{U.S. Government Accountability Office}},
  title        = {{HUD} Rental Assistance: Progress and Challenges in Measuring and Reducing Improper Rent Subsidies},
  year         = {2005},
  institution  = {U.S. Government Accountability Office},
  number       = {GAO-05-224},
  url          = {https://www.gao.gov/assets/gao-05-224.pdf},
  note         = {Accessed June 9, 2025}
}

@article{roth2002economist,
  title={The economist as engineer: Game theory, experimentation, and computation as tools for design economics},
  author={Roth, Alvin E},
  journal={Econometrica},
  volume={70},
  number={4},
  pages={1341--1378},
  year={2002},
  publisher={Wiley Online Library}
}

@article{budish2011combinatorial,
  title={The combinatorial assignment problem: Approximate competitive equilibrium from equal incomes},
  author={Budish, Eric},
  journal={Journal of Political Economy},
  volume={119},
  number={6},
  pages={1061--1103},
  year={2011},
  publisher={University of Chicago Press Chicago, IL}
}

@article{kim2024choosing,
  title={Choosing your own luck: Strategic risk taking and effort in contests},
  author={Kim, Kyungmin and Krishna, R Vijay and Ryvkin, Dmitry},
  journal={Available at SSRN 4884945},
  year={2024}
}

@article{milgrom1982theory,
  title={A theory of auctions and competitive bidding},
  author={Milgrom, Paul R and Weber, Robert J},
  journal={Econometrica},
  pages={1089--1122},
  year={1982},
  publisher={JSTOR}
}

@article{matthews1987comparing,
  title={Comparing auctions for risk averse buyers: A buyer's point of view},
  author={Matthews, Steven},
  journal={Econometrica},
  pages={633--646},
  year={1987},
  publisher={JSTOR}
}

@article{che1998standard,
  title={Standard auctions with financially constrained bidders},
  author={Che, Yeon-Koo and Gale, Ian},
  journal={The Review of Economic Studies},
  volume={65},
  number={1},
  pages={1--21},
  year={1998}
}

@article{kazumuraMishraSerizawa2020nonquasi,
  title   = {Mechanism Design without Quasilinearity},
  author  = {Kazumura, Tomoya and Mishra, Debasis and Serizawa, Shigehiro},
  journal = {Theoretical Economics},
  volume  = {15},
  number  = {2},
  pages   = {511--544},
  year    = {2020},
  month   = {May},
  doi     = {10.3982/TE2910}
}

@article{kosMessner2013nonquasi,
  title   = {Incentive Compatibility in Non-Quasilinear Environments},
  author  = {Kos, Nenad and Messner, Matthias},
  journal = {Economics Letters},
  volume  = {121},
  number  = {1},
  pages   = {12--14},
  year    = {2013},
  doi     = {10.1016/j.econlet.2013.06.020}
}

\end{document}